\PassOptionsToPackage{unicode}{hyperref}
\PassOptionsToPackage{hyphens}{url}
\PassOptionsToPackage{dvipsnames,svgnames,x11names}{xcolor}
\documentclass[
  12pt]{article}

\usepackage{iftex}
\usepackage{graphicx} 
\usepackage{amsmath,amsfonts,amssymb}
\usepackage{natbib}
\usepackage{xcolor} 
\usepackage{bm}
\usepackage[linesnumbered,ruled,vlined]{algorithm2e}
\usepackage{float} 
\usepackage{longtable}
\usepackage{booktabs}

\newtheorem{theorem}{Theorem}[section]
\newtheorem{corollary}{Corollary}[section]
 \newtheorem{lemma}{Lemma}[section]

\newcommand{\mbE}{\mathbb{E}}
\newcommand{\mbR}{\mathbb{R}}

\newcommand{\bX}{\bm{X}}

\newcommand{\bx}{\bm{x}}

\newcommand{\btheta}{\bm{\theta}}
\newcommand{\balpha}{\bm{\alpha}}

\newcommand{\bTheta}{\bm{\Theta}}
\newcommand{\bvartheta}{\bm{\vartheta}}

\newcommand{\mcG}{\mathcal{G}}
\newcommand{\mcR}{\mathcal{R}}
\newcommand{\mcH}{\mathcal{H}}
\newcommand{\mcL}{\mathcal{L}}
\newcommand{\mcN}{\mathcal{N}}
\newcommand{\mcB}{\mathcal{B}}
\newcommand{\mcX}{\mathcal{X}}
\newcommand{\mcS}{\mathcal{S}}
\newcommand{\mcD}{\mathcal{D}}
\newcommand{\mcW}{\mathcal{W}}

\ifPDFTeX
  \usepackage[T1]{fontenc}
  \usepackage[utf8]{inputenc}
  \usepackage{textcomp} 
\else 
  \usepackage{unicode-math}
  \defaultfontfeatures{Scale=MatchLowercase}
  \defaultfontfeatures[\rmfamily]{Ligatures=TeX,Scale=1}
\fi
\usepackage{lmodern}
\ifPDFTeX\else  
\fi
\IfFileExists{upquote.sty}{\usepackage{upquote}}{}
\IfFileExists{microtype.sty}{
  \usepackage[]{microtype}
  \UseMicrotypeSet[protrusion]{basicmath} 
}{}
\makeatletter
\@ifundefined{KOMAClassName}{
  \IfFileExists{parskip.sty}{%
    \usepackage{parskip}
  }{
    \setlength{\parindent}{0pt}
    \setlength{\parskip}{6pt plus 2pt minus 1pt}}
}{
  \KOMAoptions{parskip=half}}
\makeatother
\usepackage{xcolor}
\makeatletter
\ifx\paragraph\undefined\else
  \let\oldparagraph\paragraph
  \renewcommand{\paragraph}{
    \@ifstar
      \xxxParagraphStar
      \xxxParagraphNoStar
  }
  \newcommand{\xxxParagraphStar}[1]{\oldparagraph*{#1}\mbox{}}
  \newcommand{\xxxParagraphNoStar}[1]{\oldparagraph{#1}\mbox{}}
\fi
\ifx\subparagraph\undefined\else
  \let\oldsubparagraph\subparagraph
  \renewcommand{\subparagraph}{
    \@ifstar
      \xxxSubParagraphStar
      \xxxSubParagraphNoStar
  }
  \newcommand{\xxxSubParagraphStar}[1]{\oldsubparagraph*{#1}\mbox{}}
  \newcommand{\xxxSubParagraphNoStar}[1]{\oldsubparagraph{#1}\mbox{}}
\fi
\makeatother

\usepackage{longtable,booktabs,array}
\usepackage{calc} 
\usepackage{etoolbox}
\makeatletter
\patchcmd\longtable{\par}{\if@noskipsec\mbox{}\fi\par}{}{}
\makeatother
\IfFileExists{footnotehyper.sty}{\usepackage{footnotehyper}}{\usepackage{footnote}}
\makesavenoteenv{longtable}
\usepackage{graphicx}
\makeatletter
\def\maxwidth{\ifdim\Gin@nat@width>\linewidth\linewidth\else\Gin@nat@width\fi}
\def\maxheight{\ifdim\Gin@nat@height>\textheight\textheight\else\Gin@nat@height\fi}
\makeatother
\setkeys{Gin}{width=\maxwidth,height=\maxheight,keepaspectratio}
\makeatletter
\def\fps@figure{htbp}
\makeatother

\makeatletter
\@ifpackageloaded{caption}{}{\usepackage{caption}}
\AtBeginDocument{%
\ifdefined\contentsname
  \renewcommand*\contentsname{Table of contents}
\else
  \newcommand\contentsname{Table of contents}
\fi
\ifdefined\listfigurename
  \renewcommand*\listfigurename{List of Figures}
\else
  \newcommand\listfigurename{List of Figures}
\fi
\ifdefined\listtablename
  \renewcommand*\listtablename{List of Tables}
\else
  \newcommand\listtablename{List of Tables}
\fi
\ifdefined\figurename
  \renewcommand*\figurename{Figure}
\else
  \newcommand\figurename{Figure}
\fi
\ifdefined\tablename
  \renewcommand*\tablename{Table}
\else
  \newcommand\tablename{Table}
\fi
}
\@ifpackageloaded{float}{}{\usepackage{float}}
\floatstyle{ruled}
\@ifundefined{c@chapter}{\newfloat{codelisting}{h}{lop}}{\newfloat{codelisting}{h}{lop}[chapter]}
\floatname{codelisting}{Listing}

\makeatother
\makeatletter
\@ifpackageloaded{caption}{}{\usepackage{caption}}
\@ifpackageloaded{subcaption}{}{\usepackage{subcaption}}
\makeatother

\ifLuaTeX
  \usepackage{selnolig}  
\fi
\usepackage[]{natbib}
\usepackage{bookmark}

\IfFileExists{xurl.sty}{\usepackage{xurl}}{} 
\hypersetup{
  pdftitle={Title},
  pdfauthor={Author 1; Author 2},
  pdfkeywords={3 to 6 keywords, that do not appear in the title},
  colorlinks=true,
  linkcolor={blue},
  filecolor={Maroon},
  citecolor={Blue},
  urlcolor={Blue},
  pdfcreator={LaTeX via pandoc}}

\newcommand{\anon}{1}

\usepackage{hyperref}
\usepackage{xr}
\usepackage{xr-hyper} 
\begin{document}

\def\spacingset#1{\renewcommand{\baselinestretch}%
{#1}\small\normalsize} \spacingset{1}


\if1\anon
{
  \title{\bf Neural Network Machine Regression (NNMR): A Deep Learning Framework for Uncovering High-order Synergistic Effects}
  \author{Jiuchen Zhang\\
    Department of Biostatistics, University of Michigan\\
    and \\
    Ling Zhou \\
    Center of Statistical Research and School of Statistics, \\
    Southwestern University of Finance and Economics \\
    and \\
    Peter Song \\
    Department of Biostatistics, University of Michigan}
  \maketitle
} \fi

\if0\anon
{
  \bigskip
  \bigskip
  \bigskip
  \begin{center}
    {\LARGE\bf Neural Network Machine Regression (NNMR): A Deep Learning Framework for Uncovering High-order Synergistic Effects}
\end{center}
  \medskip
} \fi

\bigskip
\begin{abstract}
We propose a new neural network framework, termed Neural Network Machine Regression (NNMR), which integrates trainable input gating and adaptive depth regularization to jointly perform feature selection and function estimation in an end-to-end manner. By penalizing both gating parameters and redundant layers, NNMR yields sparse and interpretable architectures while capturing complex nonlinear relationships driven by high-order synergistic effects. We further develop a post-selection inference procedure based on split-sample, permutation-based hypothesis testing, enabling valid inference without restrictive parametric assumptions. Compared with existing methods, including Bayesian kernel machine regression and widely used post hoc attribution techniques, NNMR scales efficiently to high-dimensional feature spaces while rigorously controlling type I error. Simulation studies demonstrate its superior selection accuracy and inference reliability. Finally, an empirical application reveals sparse, biologically meaningful food group predictors associated with somatic growth among adolescents living in Mexico City.
\end{abstract}

\noindent%
\vfill

\newpage
\spacingset{1.8} 

\section{Introduction}

In many scientific fields, such as biomedical research, genomics, epidemiology, and environmental science, researchers work with high-dimensional datasets where the number of potential explanatory variables far exceeds the number of observations. Identifying a relevant subset of variables is crucial for improving model interpretability, reducing overfitting, and enhancing predictive performance. 

Variable selection has a long history in linear modeling, where sparsity penalties such as the LASSO \citep{tibshirani1996regression} and smoothly clipped absolute deviation (SCAD; \citealp{fan2001variable}) are now routine. However, many scientific questions involve nonlinear, possibly high‑order interactions that lie well beyond the scope of linear assumptions.  To accommodate such complexity, research has shifted toward nonparametric frameworks that let the data reveal flexible functional forms while still isolating the influential predictors.  Early efforts extended the linear model to additive structures in which each covariate enters through an unspecified univariate function \citep{hastie2017generalized}.  Although additive models inherit interpretability and can be equipped with component‑wise selection rules, they may miss interaction effects unless higher‑order additive terms or interaction kernels are incorporated, a step that greatly complicates both estimation and feature selection.

Complementary strategies tackle the “curse of dimensionality’’ by screening rather than penalizing.  Sure independence screening (SIS; \citealp{fan2008sure}) and its variants rank predictors through marginal utilities to select the feature space before more elaborate modeling; however, their reliance on marginal signals can overlook variables that act solely through interactions.  Fully Bayesian machinery, typified by Bayesian kernel machine regression (BKMR; \citealp{bobb2015bayesian}), integrates variable selection, nonlinear response surfaces, and uncertainty quantification in a single coherent framework, but at the expense of Markov chain Monte Carlo computation that scales poorly with thousands of samples or covariates. Collectively, these developments underscore a persistent tension in nonparametric regression: balancing modeling flexibility, statistical efficiency, computational feasibility, and interpretability.  

Deep neural networks pose a dual challenge for feature selection: the parameterization is massively over‑complete, yet the correspondence between weights and inputs is highly entangled.  Post hoc attribution tools such as SHAP \citep{lundberg2017unified}, Integrated Gradients \citep{sundararajan2017axiomatic}, and deepLIFT \citep{shrikumar2017learning} estimate ex post importance scores by propagating gradients or relevance values back to the input layer.  These methods require only a trained model and are widely used in practice, but they do not impose sparsity during training; their scores can be unstable under collinearity and provide no formal guarantee that low‑scoring variables are irrelevant—limitations that become acute in high‑dimensional biomedical applications.

To induce sparsity within neural networks, two embedded approaches have been proposed. Deep feature selection inserts a sparse linear layer at the network input and applies an $\ell_{1}$ penalty to its weights, thereby jointly learning both a representation and feature importance in an end-to-end fashion \citep{li2016deep}. Building on this framework, nonlinear variable selection methods introduce a continuous $\ell_{0}$ relaxation to selection layer, yielding rigorous convergence guarantees and selection consistency under a generalized stable Hessian condition \citep{chen2021nonlinear}.  Group Lasso regularization instead treats all outgoing connections of each input variable as a group, which prunes entire neurons or inputs to produce highly compact networks without extensive manual grouping \citep{scardapane2017group}. However, these methods still lack automated control over both network depth and width, offer limited support for valid post-selection inference.

In this paper, we propose a neural network-based variable selection method tailored to effectively identify relevant features in high-dimensional settings. Our approach incorporates several innovations that enhance interpretability, scalability, and statistical rigor. First, we introduce a trainable feature-gating mechanism at the neural network's input layer, where each feature is weighted by a learnable parameter. Coupled with an $L_1$-penalty analogous to the LASSO \citep{tibshirani1996regression}, this mechanism directly induces sparsity, explicitly identifying a concise subset of predictors most strongly associated with the outcome.

Second, our method includes an adaptive thresholding procedure during training, periodically removing features whose gating parameters remain persistently small. Additionally, we dynamically prune redundant hidden layers by replacing them with identity mappings if their contributions become negligible. These two procedures promote structured sparsity across both input features and network depth, resulting in a compact and interpretable model.

Third, we establish theoretical guarantees for our approach by deriving an risk upper bound, achieving minimax optimality under mild conditions. Combined with our proposed data-splitting and permutation-based inference framework, our method rigorously controls type-I errors in high-dimensional analyses, effectively reducing false discoveries.

Our neural network-based variable selection approach exhibits several advantages over existing methods, particularly in contexts involving large-scale biomedical and genetic data. Unlike Bayesian kernel machine regression (BKMR) \citep{bobb2015bayesian}, which involves computationally expensive kernel inversions, or traditional screening approaches \citep{fan2008sure}, our framework scales efficiently to hundreds or even thousands of predictors. Moreover, in contrast to popular post hoc interpretation methods such as SHapley Additive exPlanations (SHAP) or Integrated Gradients—often unstable or misleading in high-dimensional settings—our method directly incorporates interpretability within the model training process through structured penalties.

 Overall, our unified framework simultaneously addresses feature selection, adaptive neural network optimization, and rigorous statistical inference, positioning our method as a powerful advancement for interpretable, scalable, and statistically rigorous feature selection in contemporary statistical modeling and deep learning applications.

\section{Methodology}

We consider variable selection within the framework of neural networks. Suppose we observe independent and identically distributed samples $ \{(\boldsymbol{X}_i,Y_i)\}_{i=1}^n $, where $ Y_i \in \mathbb{R} $ is a scalar response and $ \boldsymbol{X}_i \in \mathbb{R}^d $ is a high-dimensional vector of predictors. Our objective is to simultaneously estimate the unknown relationship $ g(\cdot) $ between predictors and response, and identify a sparse subset of relevant features. To accomplish this, we embed the variable selection directly into the neural network training process. Specifically, we estimate $ g $ by minimizing the empirical squared-error loss:
\begin{equation}
\hat{g} = \arg\min_{g\in\mathcal{G}} \frac{1}{n}\sum_{i=1}^{n}(Y_i - g(\boldsymbol{X}_i))^2,
\end{equation}
where $\mathcal{G}$ denotes the class of neural networks designed explicitly to perform variable selection.

\subsection{Neural Network}
We set $ \mathcal{G} $ to be a function class consisting of multi-layers neural networks with a ReLU activation function to approximate the conditional expectation:
\[
\mathcal{G} := \mathcal{NN}(\mcW,  \mcD, \mathcal{S}, \mathcal{B}),
\]
where the input data is the predictor $ \boldsymbol{X}_i $, forming the first layer, and the output is the last layer of the network. This network $\mathcal{G} $ has width $ \mcD $, which includes $ \mcD $ hidden layers and $ \mcD+2 $ total layers. Let $ N_l $ denote the width of layer $ l $ for $ l = 0, \dots, \mcD, \mcD+1 $, which is the number of nodes in each layer. Specifically, $ N_0 = d $ represents the input dimension of $ X $, and $ N_{\mcD+1} = 1 $ represents the response $ Y $. The width $ \mcW$ is the maximum width among the hidden layers:
\[
\mcW = \max(N_1, \dots, N_{\mcD}).
\]
Without loss of generality, we consider the same width for all hidden layers in this paper. The size $ \mcS $ is the total number of parameters in the network $ \mathcal{G} $, given by $ \mcS = \sum_{l=0}^{\mcD} N_{l+1} \times (N_l + 1) $. 

In particular, a ReLU neural network with $\mcD$ hidden layers is a collection of mappings $g : \mathbb{R}^{N_0} \to \mathbb{R}^{N_{\mcD+1}}$ of the form $g( \boldsymbol{x}) = h_{\mcD} \circ \sigma \circ h_{\mcD-1} \circ \cdots \circ \sigma \circ h_0(\boldsymbol{x})$,
where $h_1 \circ h_2(\boldsymbol{x}) := h_1(h_2(\boldsymbol{x}))$ represents the composition of two functions $h_1$ and $h_2$; $\sigma(\boldsymbol{x}):= \max(\boldsymbol{x}, 0)$ is the ReLU function, which is applied elementwise; $h_l(\boldsymbol{x}) := \boldsymbol{W}_l\boldsymbol{x} + c_l$ is an affine transformation with $\boldsymbol{W}_l \in \mathbb{R}^{N_{l+1} \times N_l}$ and $c_l \in \mathbb{R}^{N_l}$. We assume that every function $ g \in \mathcal{G} $ satisfies $\lVert g \rVert_{\infty} \le \mathcal{B}$ for some $0 < \mathcal{B} < \infty$, where $\lVert g \rVert_{\infty}$ is the supnorm of the function $g$.

\begin{figure}[!h]
    \centering
    \includegraphics[]{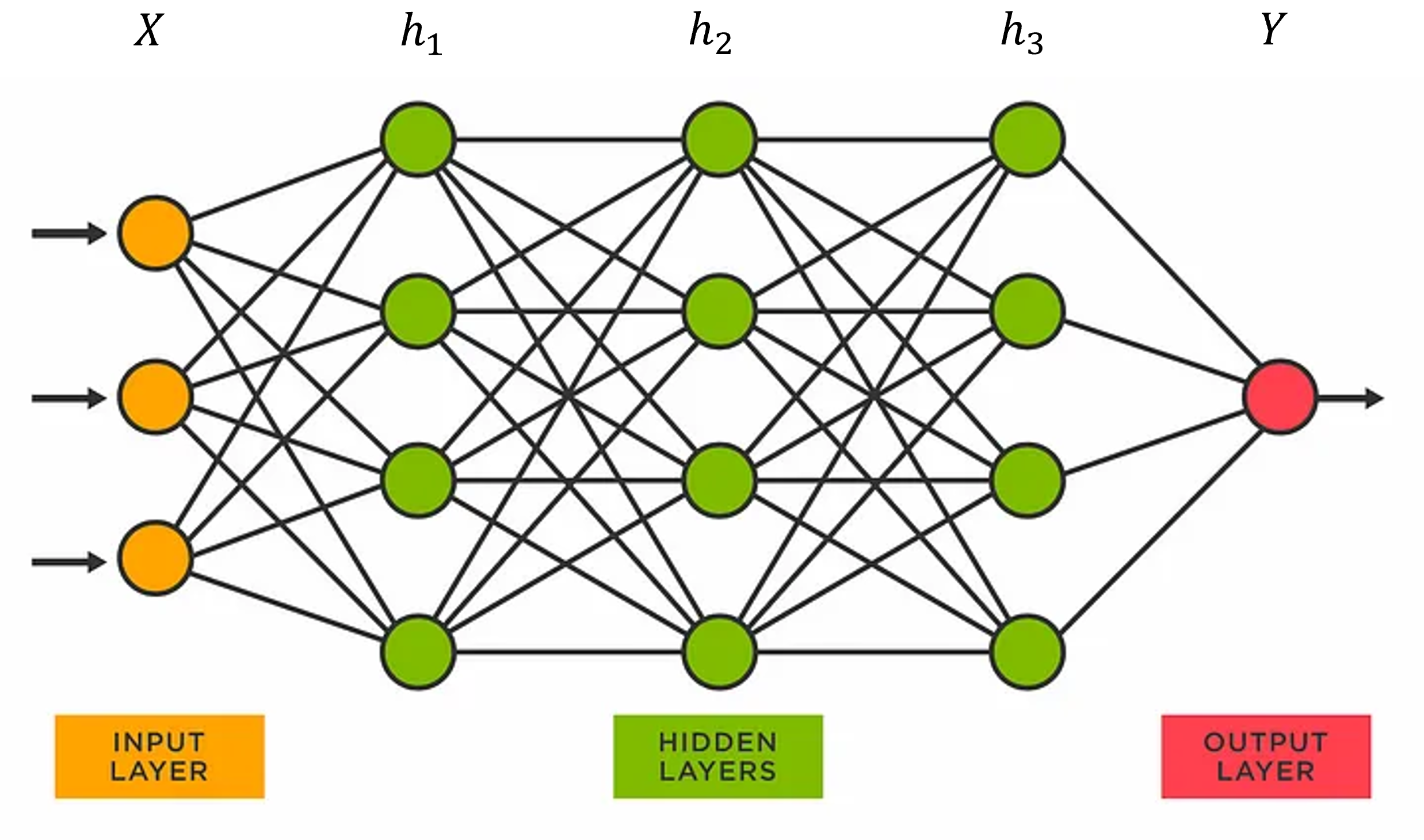}
    \caption{Simple Multi-layer Neural Network}
\end{figure}

\subsection{Variable Selection}
\label{sec:vs}

In high-dimensional nonparametric regression problems, the response variable $ Y $ often depends on only a small subset of the available features, while the rest contribute little to predictive accuracy. This suggests that the data lies on a low-dimensional manifold embedded in the high-dimensional feature space. By leveraging this assumption, we aim to identify and retain only the most relevant features, reducing computational complexity while maintaining predictive performance.

To exploit this
sparsity we attach to the input layer a trainable gating vector
$\boldsymbol{\alpha}=(\alpha_{1},\dots,\alpha_{d})^{\top}\in \mathbf{R}^{d}$ and
feed the re‑weighted input
\[
\widetilde{\boldsymbol{X}}=\boldsymbol{\alpha}\odot\boldsymbol{X},\qquad
(\boldsymbol{\alpha}\odot\boldsymbol{X})_{j}=\alpha_{j}X_{j},
\]
into the network. 

This transformation allows the model to learn which features are important by setting some entries of $ \alpha $ to zero, effectively removing those variables from consideration. The resulting feature selection mechanism is embedded within the neural network architecture, ensuring that only a subset of the input dimensions contributes to the final model.

To formalize the selection process, we incorporate $L_1$ regularization on $ \alpha $, which encourages sparsity by penalizing small coefficient values. If $ \alpha $ has exactly $ d_0 $ nonzero elements, then $\widetilde{\boldsymbol{X}}$ can be viewed as a $ d_0 $-dimensional vector, allowing the model to learn a function $ g: \mathbb{R}^{d_0} \to \mathbb{R} $. This leads to the following optimization problem:
\[
\mathcal{L}_{n}(g)
=\frac{1}{n}\sum_{i=1}^{n}
\bigl\{Y_{i}-g(\boldsymbol{X}_{i})\bigr\}^{2}
+\lambda_{1}\lVert\boldsymbol{\alpha}\rVert_{1},
\]
Here, $ g \in \mathcal{G}(\balpha,\btheta)$ is a neural network that maps the selected features to $ Y $, where $\mathcal{G}(\balpha,\btheta) = \{g: g=\phi(\boldsymbol{\alpha}\odot\boldsymbol{X}), \phi \in \mathcal{G}\}$, $\btheta = (\boldsymbol{W}_l,c_l, l= 0,\cdots,\mcD)$, and $ \| \cdot \|_1 $ serves as an $L_1$ norm that ensures sparsity in $ \balpha $. This formulation not only selects the most informative features but also learns the nonlinear mapping between those features and the response variable. Unlike traditional feature selection methods that operate in a separate preprocessing step, our approach integrates feature selection into the model training process, allowing it to dynamically adapt the selected features based on the data.

Compared to existing feature selection techniques such as Bayesian Kernel Machine Regression (BKMR) and SHAP-based methods, our approach offers significant advantages in terms of scalability and efficiency. BKMR, while powerful in modeling nonlinear interactions, suffers from high computational costs due to kernel matrix inversion, making it impractical for large-scale problems. Similarly, SHAP-based methods, which estimate feature importance post hoc, do not enforce sparsity and can be computationally expensive for high-dimensional data.

In contrast, our method is fully compatible with GPU-accelerated deep learning frameworks, allowing it to scale efficiently. The integrated feature selection mechanism eliminates the need for separate preprocessing steps, reducing overhead. Moreover, by explicitly learning a sparse representation, our approach ensures that the model remains interpretable while retaining predictive power.

The input‑gating layer reduces the width of the effective feature space, but the depth of the network remains fixed.  Retraining a new architecture each time the size of the active set changes is computationally infeasible.  Following \cite{tan2024generative}, a hidden layer $l$ is redundant when its linear map $h_{l}(x)=\boldsymbol{W}_lx$ simply forwards
its input, i.e.\ $\boldsymbol{W}_l=\boldsymbol{I}$ and $c_l = 0$, where $\boldsymbol{I}$ is the identity matrix.  Note that the bias term has been absorbed into the
augmented input.

This motivates the depth penalty,
\[
\operatorname{DP}(\theta)=\sum_{l=1}^{\mcD}(\lVert \boldsymbol{W}_l-\boldsymbol{I}\rVert_{1}+|c_l|),
\]
which is zero precisely when every dispensable layer collapses to the identity.  Combining input gating and depth regularisation, the final training problem becomes
\[
(\widehat{\balpha}, \widehat{\btheta})
=\arg\min_{g\in \mcG(\balpha, \btheta)}
\Bigl\{
\frac{1}{n}\sum_{i=1}^{n}
\bigl\{Y_{i}-g(\boldsymbol{X}_{i})\bigr\}^{2}
+\lambda_{1}\lVert\boldsymbol{\alpha}\rVert_{1}
+\lambda_{2}\sum_{l=1}^{\mcD}(\lVert \boldsymbol{W}_l-\boldsymbol{I}\rVert_{1}+|c_l|)
\Bigr\},
\tag{12}\label{eq:fullobj}
\]
with $\lambda_{1}$ controlling sparsity and $\lambda_{2}$ regulating depth.  During optimization each $\boldsymbol{W}_l$ is softly shrunk toward the identity matrix whose penalty terms converge to zero can be pruned after training, yielding a compact architecture that matches the complexity of the selected feature subset while retaining the expressive power of deep ReLU networks.

\subsection{Inference Procedure after Variable Selection}
\label{sec:inf}

Following variable selection, rigorous statistical inference is required to assess the significance of the retained predictors. We employ a data‐splitting strategy \citep{cox1975note,hurvich1990impact} combined with the model‐free test of conditional independence \citep{cai2022model}. Specifically, we partition the sample into two parts—one for both variable selection and estimation of conditional mean functions, and the other solely for inference—thereby preserving maximum power.

Specifically, we divide the original dataset into two subsets: one for variable selection and estimation (denoted $ D_1 $), and another purely for inference (denoted $ D_2 $). The first subset $D_1$ is used both to identify important predictors through the penalized neural network described in Section \ref{sec:vs} and to estimate conditional expectations for inference purposes. The inference step directly utilizes the entire second subset $D_2$ without further splitting, enhancing statistical power.

Let the selected variables identified from $D_1$ be represented as $X_S$. The inference procedure involves testing the null hypothesis:
\[
H_0: E(Y \mid \boldsymbol{X}_{-j,S}) = E(Y \mid \boldsymbol{X}_S) \quad \text{versus} \quad H_1: E(Y \mid \boldsymbol{X}_{-j,S}) \neq E(Y \mid \boldsymbol{X}_S),
\]
where $\boldsymbol{X}_{j}$ denotes potential variable of interest and $\boldsymbol{X}_{-j,S}$ denotes the rest of variables in $\boldsymbol{X}_S$ excluding $\boldsymbol{X}_j$.

Following the methodology of \citet{cai2022model} with our enhanced approach, we fit two nonparametric models using data from $D_1$: one model incorporating both $X_S$ and $Z$, denoted by $\hat{g}_1(\boldsymbol{X}_S$, and a null model that includes only $\boldsymbol{X}_{-j, S}$, denoted by $\hat{g}_0(\boldsymbol{X}_{-j, S})$. To assess the predictive performance of these models, we perform two-sample comparisons on the residuals from these models using the inference dataset $D_2$. Specifically, we utilize the two-sample t-test (TS), defined as the average difference in squared residuals between the two models. The test statistic $T_{TS}$ is computed as:
\[
T_{TS} = \frac{1}{n_2}\sum_{i \in D_2}\left[\{Y_i - \hat{g}_1(\boldsymbol{X}_{S,i})\}^2 - \{Y_i - \hat{g}_0(\boldsymbol{X}_{-j,S,i})\}^2\right],
\]
and we reject the null hypothesis $H_0$ if $T_{TS}$ is significantly negative. This test evaluates the difference in mean squared prediction errors between the two models, under the assumption that second-order moments exist for the data.


\begin{algorithm}[H]
\label{alg:infer}
\caption{Inference procedure with permutation testing}
\KwIn{Training data $D_1$, inference data $D_2$, number of permutations $B$, selected variable $\boldsymbol{X}_S$ obtained from $D_1$} 
\KwOut{Test statistic and p-value}
Estimate $\hat{g}_1(\boldsymbol{X}_S)$ and $\hat{g}_0(\boldsymbol{X}_{-j,S})$ using $D_1$\;
Evaluate both models on $D_2$ and calculate residuals: \\
$\quad U_i = Y_i - \hat{g}_1(\boldsymbol{X}_{S,i})$, \\
$\quad V_i = Y_i - \hat{g}_0(\boldsymbol{X}_{-j,S,i})$, for all $i \in D_2$\;
Define the pooled set $S = \{U_i\} \cup \{V_i\}$\;
Compute observed test statistic $T$ using $\{U_i\}, \{V_i\}$\;
\For{$b = 1$ to $B$}{
    Randomly partition $S$ into two equal-sized sets $\{U^*_i\}, \{V^*_i\}$\;
    Compute test statistic $T_{TS,b}^*$ using $\{U^*_i\}, \{V^*_i\}$\;
}
Compute p-value: $\hat{p}_{TS} = \frac{1}{B} \sum_{b=1}^{B} I\{T > T_{TS,b}^*\}$\;
\end{algorithm}

The significance of this incremental predictive power is evaluated by permutation tests, as described in Algorithm \ref{alg:infer}. Estimation of $\hat{g}_1(\boldsymbol{X}_S)$ and $\hat{g}_0(\boldsymbol{X}_{-j,S})$ using $D_1$ increases statistical efficiency and power by fully utilizing the inference dataset, ensuring robust inference without restrictive parametric assumptions. Additionally, our inference framework is flexible, allowing for individual testing of each selected variable separately, or focusing solely on specific variables of particular interest based on the research objectives.

In addition, under suitable conditions, one can show that the asymptotic Gaussianity of the t statistic $T_{TS}$ has been established in \citet{lei2020cross}. As a result, the last step of the $p$ value calculation can be modified by first estimating the standard derivation of $\{T_{TS,b}^* , b = 1, . . . , B\}$, denoted as $\hat\sigma_{B}$, and then calculating the  $p$ value as $\Phi^{-1}(T/\hat\sigma_{B})$, where $\Phi$ is the cumulative distribution function of standard normal distribution. This will provide a $p$ value with relatively high resolution.

\section{Theory}
\label{sec:theory}
In this section, we present theoretical results concerning the consistency and validity of our neural network-based variable selection approach. We begin by formally defining the theoretical framework underlying our methodology.

Let $X_{[j]}$ be the $j$-th component of $\bX$. Then $X_{[j]}$ is defined as conditionally unimportant if and only if 
\begin{eqnarray}
\label{eq:imp}
\mbE(Y \mid \bX_{[-j]} = \bx_{[-j]}) = \mbE(Y\mid \bX = \bx),
\end{eqnarray}
where $\bX_{[-j]} = (X_{[1]}, \cdots, X_{[j-1]}, X_{[j+1]}, \cdots, X_{[d]})$, and 
$\bx_{-j}$ is defined similarly. 

We assume $\bX$ is supported on a bounded set, and for simplicity, we assume this bounded set to be $[0, 1]^d$. In the rest of the paper, the constant $c$ denotes a positive constant that may vary across different contexts. We separate $\bX$ into $\bX^s \in [0, 1]^{d_0}$ and $\bX^c \in [0, 1]^{d-d_0}$, denoting the conditional important and 
unimportant variables, respectively. For any $\bx \in [0, 1]^p$, 
$\mbE(Y \mid \bX^{s} = \bx^{s}) = \mbE(Y\mid \bX = \bx)$. 

Let $g^*$ be the minimizer of the population loss, that is, 
\[
g^* = \arg\min_{g \in \mcH^{\beta}([0, 1]^p,  \mcB)} \mbE\left( Y_i - g(\bX_i) \right)^2,
\]
where the minimizer is taken over the entire space and thus implies that $g^*$ does not necessarily belong to the FNN set $\mcG(\balpha, \btheta)$. Clearly, based on the definition~\eqref{eq:imp} $g^*(\bX)$ depends on $\bX^s$ only.  To simplify notation, we assume the first $d$ components of $\bX$ are important, and divide $\balpha$ into two parts $\balpha = (\balpha_s, \balpha_c)$ with  $\balpha_s \in \mbR^d$ and $\balpha_c \in \mbR^{p-d}$ corresponding to the coefficients of $\bX^s$ and $\bX^c$, respectively. 

We establish the large sample property of $\hat{g}$ in terms of its excess risk, which is defined as 
the difference between the risk of $g$ and $g^*$:
\[
\mcR(g) - \mcR(g^*) = \mbE\left( Y_i - g(\bX_i) \right)^2 - \mbE\left( Y_i - g^*(\bX_i) \right)^2.
\]

We construct an estimator of $g^*$ within $\mcG(\balpha, \btheta)$ with network parameters $\btheta$ and variable importance index parameters $\balpha$. Thus, referring to the optimal network solution 
\[
g^*_{\mcG} = \arg\min_{g \in \mcG(\balpha, \btheta)}\mbE\left( Y_i - g(\bX_i) \right)^2, 
\]
we define a set of $(\balpha, \btheta)$ as 
\[
\bTheta = \{(\balpha, \btheta): \mcR(g_{\balpha, \btheta}) = \mcR(g^*_{\mcG}), g_{\balpha, \btheta} \in \mcG(\balpha, \btheta)\}.
\]
With proper conditions on the data structure and network class, $g^*_{\mcG}$ can approach $g^*$ close enough such that 
$\balpha_c = 0$ and $\balpha_s \neq 0$. That is, for any $(\balpha, \btheta) \in \bTheta$, there exists a positive constant $c > 0$ such that 
$\min_{1 \leq j  \leq d}|\alpha_j| \geq c$. Further, there exists a solution $\balpha_s \neq 0$, $\balpha_c = \bm{0}$ such that $(\balpha_s, \bm{0}, \btheta) \in \bTheta$ still holds.  These properties of 
$g^*_{\mcG}$ and $\bTheta$ are given in Lemma~\ref{th:le1} in the Supplementary Materials. 

In practice, the network parameters $(\hat{\balpha}, \hat{\btheta})$ obtained in~\eqref{eq:fullobj} lies outside of $\bTheta$. To measure the distance between 
$(\balpha, \btheta)$ and $\bTheta$, we define metric 
$d((\balpha, \btheta), \bTheta) := \min_{\bvartheta \in \bTheta}\|\binom{\balpha}{\btheta} - \bvartheta\|^2_2$.
Clearly, $d((\hat{\balpha}, \hat{\btheta}), \bTheta) \to 0$ is sufficient for $\hat{\balpha}_s$ to be bounded away from zero. Thus, we can 
establish the selection results via establishing the convergence rate of the proposed penalized network $\hat{g}_{\mcG}$ in terms of its 
parameters. 

We define $\mcG|_{\bx} := \{g(\bx_1), g(\bx_2), \cdots, g(\bx_n): g \in \mcG(\balpha, \btheta)\}$ for a given sequence $\bx = (\bx_1, \cdots, \bx_n)$ and denote $\mcN_{2n} = \sup_{\bx}\mcN_{2n}(\delta, \|\cdot\|_{\infty}, \mcG|_{\bx})$ as the covering number of $\mcG|_{\bx}$ under the norm $\|\cdot\|_{\infty}$ with radius $\delta$. Let $A \preceq B$ represent $A \leq c B$ for a positive constant $c$, $A \wedge B=\min(A, B)$, and $\Phi_1 \circ \Phi_2:= \{
\phi_1 \circ \phi_2: \phi_1 \in \Phi_1, \phi_2 \in \Phi_2 \}$ represents the composition of two function classes $\Phi_1$ and $\Phi_2$.

The next conditions are needed to establish the theoretical results:

\begin{itemize}
\item[{(C1)}] The dimension of conditional important variables $d$ is fixed, and there exists a constant 
$\tau_d > 0$ such that 
\[
\min_{1 \leq j \leq d} \bigg| \mbE(Y \mid \bX_{[-j]} = \bx_{[-j]}) - \mbE(Y \mid \bX = \bx) \bigg| \geq c\tau_d,
\]
for some positive constant $c$. 
\item[{(C2)}] Assume $Y$ is a sub-Gaussian random variable. 
\item[{(C3)}] Function class for $g_{\balpha, \btheta}$ and $g^*$: for any function $g_{\balpha, \btheta} \in \mcG$ and the true function $g^*$, we assume 
$\|g_{\balpha, \btheta}\|_{\infty} < \mcB$ and $\|g^*\|_{\infty} < \mcB$.
\end{itemize}

\begin{theorem}
\label{th:cons}
Suppose the conditions (C1)-(C3) hold. If $\lambda^2_1 \preceq \frac{\log^2 n \log \mcN_{2n}}{nd}  + 
 \frac{\left( \mcR(g^*_{\mcG}) - \mcR(g^*) \right)}{d}$ and $\lambda_2^2  \preceq \frac{\log^2 n \,\log \mcN_{2n}}{n\mcW^2\mcD}
\;+\;\frac{\big(\mcR(g_{\check\balpha,\check\btheta})-\mcR(g^*)\big)}{\mcW^2\mcD}$, 
$(\hat{\balpha}, \hat{\btheta})$ defined in~\eqref{eq:fullobj} satisfies 
\begin{eqnarray*}
\mbE[d((\hat{\balpha}, \hat{\btheta}), \bTheta)] &\preceq&  \frac{\log^2 n \log \mcN_{2n}}{n} +  
\log^2 n \left( \mcR(g^*_{\mcG}) - \mcR(g^*) \right)\\
\mbE\|\hat\balpha_c\|^2_2 &\preceq& \frac{\log^2 n \log \mcN_{2n}}{n} +  \log^2 n \left( \mcR(g^*_{\mcG}) - \mcR(g^*) \right),
\end{eqnarray*}
where $\mbE$ is taken with respect to $(\bX_i, Y_i)_{i=1}^n$. 

\end{theorem}

Now, we further explore how the error relies on the FNN structure and  the function class to which $g^*$ belongs.
We consider the H{\"o}lder class, which is broad enough to cover most applications. In particular, denote $\left \lceil a \right \rceil$ and $\left \lfloor a \right \rfloor$ to be  the smallest integer no less than $a$ and the largest integer strictly smaller than $a$, respectively. Let $\mathbb{N}^+$ be the set of positive integers and $\mathbb{N}_0$ be the set of nonnegative integers.
 Let $\beta = s + r$, $r\in (0,1]$ and $s = \left \lfloor \beta \right \rfloor \in \mathbb{N}_0$. 
For a finite constant $B_0>0$, the $H\ddot o lder$ class $\mathcal{H}_\beta([0,1]^d,B_0)$ is defined as
\begin{eqnarray*}
\begin{aligned}
    		\mathcal{H}_\beta([0,1]^d,B_0)=\{&g:[0,1]^d \mapsto \mathbb{R},\max_{\| \alpha \|_1 < s} \| \partial^{\alpha} g \|_{\infty} \leq B_0, \\
            &\max_{\| \alpha \|_1 = s} \sup_{x \neq y} \frac{\left | \partial^{\alpha} g(x) - \partial^{\alpha} g(y) \right |}{\| x-y \|^r_2} \leq B_0       \},
\end{aligned}
	\end{eqnarray*}
where $\partial^\alpha = \partial^{\alpha_1} \cdots \partial^{\alpha_d}$ with $\alpha=(\alpha_1,\cdots,\alpha_d)^{\top} \in \mathbb{N}_0^d$ and $\| \alpha \|_1 = \sum_{i=1}^{d}|\alpha_i|$.

Based on Theorem 3.3 of \citet{jiao2023deep} for the approximation error in terms of FNN structures and Theorem 3 and 7 of \citet{bartlett2019nearly} for the bounding covering number, we can conclude
the following Corollary~\ref{cor:rate} from Theorem~\ref{th:cons}:
\begin{corollary}
\label{cor:rate}
Given $H\ddot o lder$ smooth functions $g^* \in \mathcal{H}_\beta([0,1]^d,B_0)$, for any $D \in \mathbb{N}^+$, $W \in \mathbb{N}^+$, under conditions of Theorem~\ref{th:cons}, conditions of Theorem 3.3 in \citet{jiao2023deep} and Theorem 3 and 7 in \citet{bartlett2019nearly}, if the FNN with a ReLU activation function has width $\mathcal{W} = c (\left \lfloor \beta \right \rfloor +1)^{2}d^{\left \lfloor \beta \right \rfloor+1}W\left \lceil \log_2(8W) \right \rceil$ and depth $\mathcal{D} = c(\left \lfloor \beta \right \rfloor +1)^{2}D\left \lceil \log_2(8D) \right \rceil$ and $\lambda_1 \preceq n^{-1}\mathcal{S}\mathcal{D} \log (\mathcal{S})/d  + (WD)^{-4\beta/d}/d$ and $\lambda_2 \preceq n^{-1}\mathcal{S}\log (\mathcal{S})/\mcW^2  + (WD)^{-4\beta/d}/(\mcW^2\mcD)$, then
\begin{eqnarray*}
\mbE[d((\hat{\balpha}, \hat{\btheta}), \bTheta)] &\preceq& 
(\log^2 n)n^{-1}\mathcal{S}\mathcal{D} \log (\mathcal{S})  + (WD)^{-4\beta/d} (\log^2 n). \\
\mbE\|\hat\balpha_c\|^2_2 &\preceq& (\log^2 n)n^{-1}\mathcal{S}\mathcal{D} \log (\mathcal{S})  + (WD)^{-4\beta/d} (\log^2 n).
\end{eqnarray*}
\end{corollary}
To facilitate reading, Theorem 3.3 of \citet{jiao2023deep} and Theorems 3 and 7 of \citet{bartlett2019nearly} are also shown in Lemmas~ and~ in Appendix. In Corollary~\ref{cor:rate}, the first term comes from the covering number of $\mcG$, which is bounded by its VC dimension
$\log \mcN_{2n}(n^{-1}, \|\cdot\|_{\infty}, \mcG|_{\bx}) = O(\mcS \mcD \log (\mcS/n^{-1}))$ \citep{bartlett2019nearly}, where $\mcS$ and $\mcD$ are the total number of parameters and hidden layers, respectively. The second term follows from the approximation results from \citet{jiao2023deep} that
$\left \|g^* - g_{\mcG}^*\right \|_{\infty} \leq 18 B_0(\left \lfloor \beta \right \rfloor +1)^{2}d^{\left \lfloor \beta \right \rfloor+\max\{ \beta, 1 \}/2}(WD)^{-2\beta/d}$ and $\mbE(\mcR(g_{\mcG}^*) - \mcR(g^*)) \simeq \mbE|g_{\mcG}^* - g^*|^2$, where $A \simeq B$ represents $A \preceq B$ and $B \preceq A$. 

\begin{corollary}
\label{th:cor}
Suppose the conditions of Corollary~\ref{cor:rate} hold. If $\lambda_1 \preceq n^{-1}\mathcal{S}\mathcal{D} \log (\mathcal{S})/d  + (WD)^{-4\beta/d}/d$ and $\lambda_2 \preceq n^{-1}\mathcal{S}\log (\mathcal{S})/\mcW^2  + (WD)^{-4\beta/d}/(\mcW^2\mcD)$, for any 
$j = 1, \cdots, d$ and $k = d+1, \cdots, p$, it holds that 
\[
\text{Pr}\left( \{ |\hat\alpha_j| \geq \tau_d\} \cap \{ |\hat\alpha_k| \leq \tau_d \} \right) \to 1,
\]
where $\tau_d$ is defined in Condition~(C1).
\end{corollary}

\section{Computation}

Optimising the composite objective
\[
\mathcal{J}(g,\boldsymbol{\alpha})
=\frac{1}{n}\sum_{i=1}^{n}\bigl\{Y_{i}-g(\boldsymbol{\alpha}\odot\boldsymbol{X}_{i})\bigr\}^{2}
+\lambda_{1}\lVert\boldsymbol{\alpha}\rVert_{1}
+\lambda_{2}\sum_{l=1}^{\mcD}(\lVert \boldsymbol{W}_l-\boldsymbol{I}\rVert_{1}+|c_l|)
\tag{11}
\]
is challenging because the $L_{1}$ penalties on $\boldsymbol{\alpha}$ and on each $\boldsymbol{W}_l-\boldsymbol{I}$ are nondifferentiable at zero, yet we need exact zeros to identify both irrelevant inputs and redundant layers.  We therefore employ a periodic hard‐thresholding (truncation). Continuous shrinkage via these gradient descent steps ensures that small coefficients approach zero, while exact zeros are enforced by truncation every $K$ iterations.

In the truncation step, any $\alpha_{j}$ satisfying $|\alpha_{j}|\le\tau_{1}$ is set to zero, and any layer $l$ for which $\|\boldsymbol{W}_l-\boldsymbol{I}\|_{1}\le\tau_{2}$ is collapsed by resetting $\boldsymbol{W}_l=\boldsymbol{I}$.  Threshold $\tau_{1}$ is selected based on validation performance, and we follow \citet{scardapane2017group} in fixing $\tau_{2}=10^{-2}$ for synthetic studies and $10^{-3}$ for real data.  By combining subgradient shrinkage with periodic truncation, the algorithm achieves stable convergence, exact sparsity in both $\boldsymbol{\alpha}$ and network depth, and full compatibility with GPU‐accelerated training pipelines.

\begin{algorithm}[H]
\caption{Truncation procedure}\label{alg:trunc}
\DontPrintSemicolon
\SetKwInOut{Input}{Input}\SetKwInOut{Output}{Output}
\Input{gate vector $\boldsymbol{\alpha}$; weight matrices $\{\boldsymbol{W}_l\}$;
thresholds $\tau_{1},\tau_{2}$}
\Output{$\boldsymbol{\alpha}_{\mathrm{trunc}},\{\boldsymbol{W}_l^{\mathrm{trunc}}\}$}
$\boldsymbol{\alpha}_{\mathrm{trunc}}\gets\boldsymbol{\alpha}$\;
\For{$j=1$ \KwTo $d$}{
  \If{$\lvert\alpha_{j}\rvert\le\tau_{1}$}{$\alpha_{j}\gets0$}
}
\For{$l=1$ \KwTo $D$}{
  \If{$\lVert \boldsymbol{W}_l-\boldsymbol{I}\rVert_{1}\le\tau_{2}$}{$\boldsymbol{W}_l\gets \boldsymbol{I}$}
}
\Return $\boldsymbol{\alpha}_{\mathrm{trunc}},\{\boldsymbol{W}_l^{\mathrm{trunc}}\}$
\end{algorithm}

\section{Simulations}
\label{sec:sim}

We conduct simulation studies to evaluate the empirical performance of the proposed neural network-based variable selection method. In particular, we assess its ability to identify the true set of active predictors and compare it against several widely used approaches: Bayesian Kernel Machine Regression (BKMR), SHapley Additive exPlanations (SHAP), and DeepLIFT. The evaluation focuses on variable selection accuracy under a complex nonlinear model with sparse signal.

\subsection{Consistency of Variable Selection}

Data are generated from the nonlinear model
\[
Y = X_{0}^{3}\,(X_{2}^{2}+X_{5}) \;-\; |X_{7}|\,\cos(X_{8}) \;+\;\varepsilon,
\]
where $\varepsilon\sim N(0,1)$ and each predictor $X_j\sim N(0,1)$ independently for $j=0,\dots,199$.  Only five coordinates $\{0,2,5,7,8\}$ enter the truth, while the remaining $d-5=195$ variables are noise.  We take $n=1000$ samples and repeat the entire simulation $100$ times to assess variability.

We compare our method against five established approaches. Bayesian kernel machine regression (BKMR) \citep{bobb2015bayesian} fits a Gaussian-process-type model and uses posterior inclusion probabilities for variable selection. SHapley additive explanations (SHAP) \citep{lundberg2017unified} compute Shapley values post hoc on a pretrained neural network to rank feature importance. DeepLIFT \citep{shrikumar2017learning} is a backpropagation-based attribution method that compares activations to a reference and assigns contribution scores to each input. GLNN \citep{scardapane2017group} applies group Lasso regularization to neural networks, penalizing groups of weights. Deep feature selection (DPS) \citep{li2016deep,chen2021nonlinear} incorporates sparsity constraints via selection layers embedded in deep architectures. For SHAP and DeepLIFT, we use the same network architecture as our method but without additional adaptive thresholding or depth adjustment, while for GLNN and DPS, we use the architectures specified in their original papers. For all competing methods except BKMR and the proposed method, the top five features are selected by weight magnitude.

Selection accuracy is quantified by the following metrics:
\[
\text{Precision}
= \frac{\lvert \hat{S}_n \cap S^\star\rvert}{\lvert \hat{S}_n\rvert},
\quad
\text{Recall}
= \frac{\lvert \hat{S}_n \cap S^\star\rvert}{\lvert S^\star\rvert},
\quad
F_1
= \frac{2\,\text{Precision}\times \text{Recall}}{\text{Precision} + \text{Recall}}.
\]
Table~\ref{tab:selection-performance} reports the mean and standard deviation of these metrics over 100 simulation replicates.

\begin{table}[H]
\centering
\caption{Variable selection performance across competing methods.}
\label{tab:selection-performance}
\begin{tabular}{lccc}
\hline
\hline
Method     & Precision      & Recall         & F1 Score       \\
\hline
Proposed   & 0.927 (0.150)  & 0.880 (0.183)  & 0.879 (0.135)  \\
BKMR       & 0.025 (0.000)  & 1.000 (0.000)  & 0.049 (0.000)  \\
SHAP       & 0.780 (0.060)  & 0.780 (0.060)  & 0.780 (0.060)  \\
DeepLIFT   & 0.720 (0.133)  & 0.720 (0.133)  & 0.720 (0.133)  \\
GLNN & 0.460 (0.100)  & 0.460 (0.100)  & 0.460 (0.100)  \\
DPS        & 0.244 (0.083)  & 0.244 (0.083)  & 0.244 (0.083)  \\
\hline
\hline
\end{tabular}
\end{table}

As shown in Table~\ref{tab:selection-performance}, our proposed method achieves an $F_1$ score of approximately 0.88, significantly outperforming competing approaches. BKMR, despite achieving perfect recall, exhibits extremely low precision and consequently poor overall selection performance ($F_1 \approx 0.05$), indicating severe over-selection. SHAP and DeepLIFT yield moderate selection accuracy but lag behind our proposed method, likely due to their post hoc attribution nature that ignores structured sparsity during training. GLNN and DPS show considerably weaker performance, reflecting limitations in selecting highly relevant nonlinear features under limited depth or overly restrictive sparsity constraints. These findings underscore the advantage of embedding adaptive sparsity constraints and automatic depth selection directly into the neural network training process, thereby substantially enhancing variable selection accuracy and interpretability in high-dimensional, nonlinear modeling scenarios.

\subsection{Post‐selection Inference}

After demonstrating the consistency of our variable selection procedure, we examine its impact on post‐selection inference.  To this end, we test the null hypothesis
\[
H_0:\;E(Y\mid X_{1},Z)=E(Y\mid Z),
\]
where $Z$ denotes the set of covariates conditioned upon.  In the “no‐selection’’ scenario, $Z$ comprises all predictors except $X_{1}$; in the post‐selection scenario, $Z$ includes only those features retained by our variable‐selection procedure (excluding $X_{1}$).  We then apply the permutation‐based test outlined in Section \ref{sec:inf} to the inference data $D_2$ under both scenarios, and report the empirical Type I error rates in Table~\ref{tab:type1}.

\begin{table}[H]
\centering
\caption{Empirical Type I error rates for testing the null hypothesis $H_0: E(Y \mid X_{1},Z) = E(Y \mid Z)$ with sample size $n=1000$ and 500 simulation replicates. Results are shown for inference conducted with and without a preceding variable selection step.}
\label{tab:type1}
\begin{tabular}{l|c}
\hline
\hline
Inference setting    & Type I error \\
\hline
No selection         & 0.13 \\
Post‐selection       & 0.01 \\
\hline
\hline
\end{tabular}
\end{table}

Table~\ref{tab:type1} demonstrates that inference without prior selection inflates the Type I error to 13\%, whereas conditioning on the data‐driven subset restores nominal control at 1\%.  This confirms that our post‐selection inference procedure maintains valid error rates even after feature selection.  

\section{Real Data Application}
\label{sec:real}

The real data analysis uses the Early Life Exposure to Environmental Toxicants (ELEMENT) Project. ELEMENT consists of three mother–child cohorts recruited in Mexico City from 1994 to 2005. More than 2,000 women and their children were enrolled and followed from pregnancy through adolescence. We apply these data to evaluate a nonparametric Neural Network Machine Regression (NNMR) method with statistical inference. Our main goal is to identify which prenatal exposures are most strongly associated with infant growth. Infant growth is measured by WHO anthropometric z-scores considering age and sex. Measurements are collected at regular visits from birth to age five, as well as during cholesterol substudies and at the P20 follow up. The analysis includes a high dimensional panel of dietary patterns grouped by food categories for both mothers and children. We exclude records with missing or implausible z-scores and standardize all exposure variables before modeling.

We evaluated the performance of our Neural Network Machine Regression (NNMR) method for variable selection in high dimensional dietary data, using infant growth z-scores as the outcome. We compared NNMR to the competing method mentioned in section \ref{sec:sim}. All methods were applied to a dataset containing child and mother dietary exposures, with performance evaluated by Akaike Information Criterion (AIC) calculated on separate training and test datasets.

\begin{table}[htbp]
\centering
\caption{Average AIC and computation time (seconds) for six competing methods applied to the ELEMENT dietary exposure data. Reported values are means with standard deviations in parentheses, calculated on test sets across 100 replicates.}
\label{tab:aic}
\begin{tabular}{c|c|c}
\hline\hline
Method     & AIC      & Time (s)\\
\hline
NNMR       & 2164.98 (23.66)   & 6.74 (1.20) \\
BKMR       & 2230.27 (23.71)   & 78.21 (20.31) \\
SHAP       & 2166.99 (25.33)   & 87.01 (8.86) \\
DeepLIFT   & 2162.45 (23.33)   & 3.66 (0.56) \\
GroupLASSO & 2195.56 (41.16)   & 2.68 (0.43) \\
DFS        & 2167.62 (23.36)   & 7.69 (1.59) \\
\hline\hline
\end{tabular}
\end{table}

Table~\ref{tab:aic} shows that NNMR achieves competitive AIC values and computation times compared with DeepLIFT, while outperforming SHAP, DFS, and Group LASSO, and substantially improving upon BKMR. The poorer performance of BKMR reflects its tendency to select excessively large models, often including nearly all variables. DeepLIFT and SHAP yield prediction errors comparable to NNMR, since they share the same underlying model, but DeepLIFT tends to select the smallest variable set. Group LASSO exhibits unstable behavior, either selecting nearly all variables or only one or two. By contrast, NNMR demonstrates modest variability in AIC across splits, highlighting its stable generalization ability. 

\begin{figure}[!h]
    \centering
    \includegraphics[]{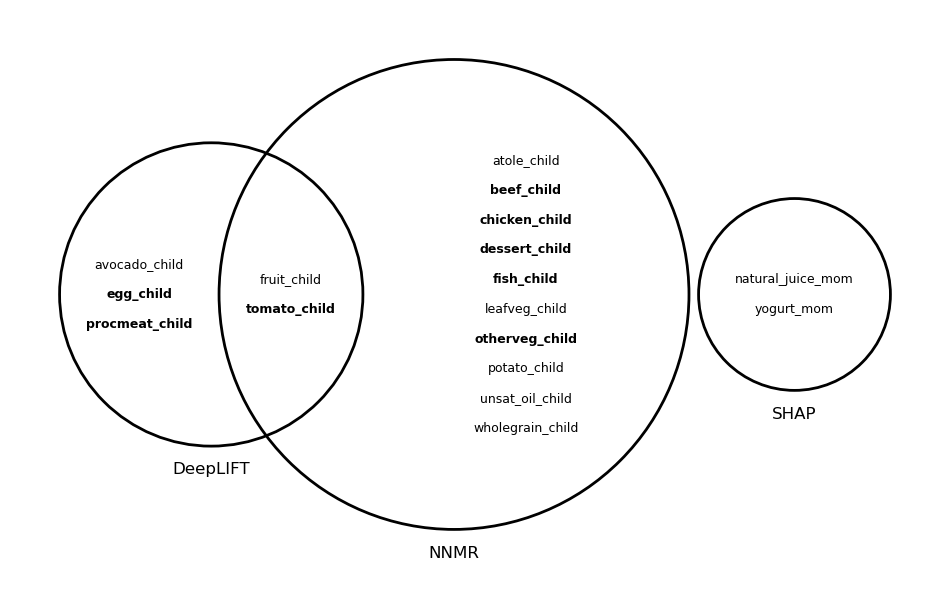}
    \caption{Venn diagram of dietary exposures selected by NNMR, SHAP, and DeepLIFT in the ELEMENT study. Displayed are variables with selection frequency exceeding 20\% across 100 replicates. Variables in bold indicate those that remained statistically significant after applying the post-selection inference procedure.}
    \label{fig:vs}
\end{figure}

Figure~\ref{fig:vs} summarizes variable selection via a Venn diagram, comparing NNMR, SHAP, and DeepLIFT. To improve interpretability, we display only variables selected with frequency exceeding 20\% across 100 replicates. The complete results are provided in the Supplementary Materials. Consistent with prior observations, BKMR selects nearly all exposures, while Group LASSO exhibits unstable behavior. DFS selects only \texttt{milk\_child}, while SHAP emphasizes maternal food groups. In contrast, NNMR and DeepLIFT consistently highlight child dietary exposures with higher and more coherent selection frequencies.

We then applied the inference procedure in Section~\ref{sec:inf} to the union of variables identified by NNMR, DFS, DeepLIFT, and SHAP, using an inference split not employed for selection. Variables highlighted in red in Figure~\ref{fig:vs} were found to be statistically significant predictors of growth z-scores, the majority of which were prioritized by NNMR. 

The five variables uniquely retained by NNMR after inference align closely with established findings in the nutritional epidemiology literature. Higher intake of beef, zinc, and choline during infancy has been linked to improved inhibitory control and attention at ages 3 to 5 years \citep{wilk2022early}. Chicken production and consumption in nutrition sensitive agricultural programs have been shown to benefit child growth in low income settings \citep{passarelli2020chicken}. Fish serves as a nutrient-dense protein source, rich in vitamins, minerals, and essential fatty acids \citep{yilmaz2018importance}. Frequent dessert consumption is associated with increased risk of overweight, a major determinant of growth z-scores \citep{barroso2016food}. Finally, the \texttt{otherveg} category, which includes zucchini, cucumber, and green beans, reflects the importance of dietary diversity for child development \citep{arimond2004dietary, thorne2019dietary}. Collectively, these results underscore the biological plausibility of the predictors identified by NNMR and their relevance to early life growth. 

\section{Discussion}
We have introduced Neural Network Machine Regression (NNMR), an integrated neural network framework designed for simultaneous feature selection, nonlinear function estimation, and rigorous post-selection inference in high-dimensional data analysis. By embedding a trainable gating layer coupled with an $L_1$ regularization strategy and an adaptive thresholding mechanism, NNMR achieves direct sparsity enforcement in input features and hidden network layers, resulting in highly compact and interpretable models.

Our theoretical analysis provides guarantees for consistent recovery of the relevant feature set, supported by a minimax-optimal upper bound on the risk under mild assumptions. Additionally, our integrated split-sample permutation testing approach ensures robust control over type I error rates, mitigating false-positive risks commonly encountered in high-dimensional inference.

Empirical results from extensive simulations illustrate that NNMR outperforms established approaches such as BKMR, SHAP, DeepLIFT, GLNN, and DPS in terms of variable selection precision, recall, and overall accuracy (F$_1$ score). While BKMR achieves high recall at the cost of severe over-selection, and SHAP and DeepLIFT offer moderate accuracy, NNMR clearly demonstrates superior performance in identifying true predictors and controlling selection errors. GLNN and DPS exhibit considerably lower accuracy, underscoring the limitations of traditional sparsity approaches in capturing complex nonlinear relationships.

Real-data applications further validate NNMR’s practical utility by demonstrating its capability to pinpoint meaningful, interpretable predictors in complex biomedical datasets. Unlike post hoc attribution methods, which are often unstable under feature collinearity, NNMR’s built-in sparsity ensures accurate identification of genuinely inactive features. The dynamic pruning of unnecessary layers also facilitates resource-efficient implementations suitable for GPU acceleration.

In summary, NNMR bridges interpretability and predictive modeling power by embedding statistically rigorous variable selection, adaptive model optimization, and valid inference within a unified training framework. This integrative approach positions NNMR as a robust and versatile solution for modern high-dimensional statistical and deep learning applications.

\bibliography{reflist.bib}
\section*{Appendix}
\subsection*{Full List of Variable Selection Results in ELEMENT Study}

Table~\ref{tab:fullvars_filtered} presents the complete list of variables along with their selection frequencies across repeated data splits for the proposed Neural Network Machine Regression (NNMR), SHAP, DeepLIFT, and DFS methods. For SHAP, only variables with selection frequencies greater than 0.05 are displayed to improve readability. Since BKMR and GroupLASSO consistently selected nearly all variables with frequencies close to 1.0 and 0.45, their results are omitted from the table for brevity. This presentation highlights the variables that were most consistently identified as important by the methods that demonstrated more selective behavior.

\begin{longtable}{lrrrr}
\caption{Variables with nonzero selection frequencies across repeated splits for NNMR, SHAP (only $\ge$ 0.05 shown), DeepLIFT, and DFS. BKMR and GroupLASSO selected nearly all variables (frequencies $\approx$ 1.0 and 0.45) and are omitted for brevity.} \label{tab:fullvars_filtered} \\
\toprule
feature & NNMR & SHAP & DeepLIFT & DFS \\
\midrule
\endfirsthead
\\
\toprule
feature & NNMR & SHAP & DeepLIFT & DFS \\
\midrule
\endhead
\midrule
\multicolumn{5}{r}{Continued on next page} \\
\midrule
\endfoot
\bottomrule
\endlastfoot
$atole\_child$ & 0.38 & -- & -- & -- \\
$avocado\_child$ & 0.15 & -- & 0.25 & -- \\
$beef\_child$ & 0.41 & -- & 0.19 & -- \\
$chicken\_child$ & 0.43 & -- & 0.15 & -- \\
$chili\_child$ & 0.09 & -- & -- & -- \\
$chips\_child$ & 0.02 & -- & -- & -- \\
$corn\_tortilla\_child$ & 0.01 & 0.08 & -- & -- \\
$corncob\_child$ & 0.13 & -- & -- & -- \\
$corncob\_mom$ & -- & 0.06 & -- & -- \\
$cruveg\_child$ & 0.01 & -- & -- & -- \\
$dessert\_child$ & 0.31 & 0.06 & -- & -- \\
$egg\_child$ & 0.16 & -- & 0.30 & -- \\
$fish\_child$ & 0.23 & -- & 0.03 & -- \\
$fish\_mom$ & -- & 0.06 & -- & -- \\
$fruit\_child$ & 0.38 & -- & 0.74 & -- \\
$fruit\_mom$ & -- & 0.07 & -- & -- \\
$hf\_dairy\_child$ & 0.01 & -- & 0.16 & -- \\
$jam\_child$ & 0.01 & -- & -- & -- \\
$leafveg\_child$ & 0.39 & -- & -- & -- \\
$legumes\_child$ & 0.03 & -- & -- & -- \\
$milk\_child$ & 0.02 & -- & 0.19 & 1.00 \\
$natural\_juice\_mom$ & -- & 0.48 & -- & -- \\
$organmeat\_child$ & 0.16 & -- & 0.03 & -- \\
$otherveg\_child$ & 0.35 & -- & -- & -- \\
$pork\_child$ & 0.04 & -- & 0.09 & -- \\
$potato\_child$ & 0.29 & -- & 0.04 & -- \\
$procmeat\_child$ & 0.02 & 0.06 & 0.23 & -- \\
$refgrain\_child$ & 0.13 & 0.12 & -- & -- \\
$refgrain\_mom$ & -- & 0.08 & -- & -- \\
$soup\_child$ & 0.18 & -- & -- & -- \\
$sugar\_beverages\_child$ & 0.05 & 0.07 & -- & -- \\
$sugar\_beverages\_mom$ & -- & 0.13 & -- & -- \\
$tomato\_child$ & 0.39 & -- & 0.22 & -- \\
$unsat\_oil\_child$ & 0.43 & -- & -- & -- \\
$wholegrain\_child$ & 0.26 & -- & -- & -- \\
$yeveg\_child$ & 0.02 & 0.09 & -- & -- \\
$yogurt\_child$ & 0.01 & -- & 0.15 & -- \\
$yogurt\_mom$ & -- & 0.32 & -- & -- \\
\end{longtable}

\subsection*{Proofs of Theorems}
Denote $\mcR_n(g_{\balpha, \btheta}) = \frac{1}{n}\sum_{i=1}^n\big( Y_i - g(\bX_i) \big)^2$, for $g \in \mcG(\balpha,\btheta)$. For any independent and identically distributed (i.i.d.) samples $D_n = \{\bX_i, Y_i\}_{i=1}^n$ with sample size $n$. 
Define 
\[
S(g_{\balpha, \btheta}, \bX_i) = \big(Y_i - g( \bX_i)\big)^2 - \big(Y_i - g^*(\bX_i)\big)^2.
\]
Let $D_n' = \{\bX_i^{\prime}, Y^{\prime}_i\}$ be another sample independent of $D_n$, and write 
\[
L(g_{\balpha, \btheta}, \bX_i^{\prime}) = \mbE_{D'_n}\left(S(g_{\balpha, \btheta}, \bX_i^{\prime}) \right)
 - 2S(g_{\balpha, \btheta}, \bX_i).
 \] 
 Recall that $g_{\btheta}(\balpha \odot \cdot) \in \mcG(\balpha, \btheta)$ and 
 \[
(\hat{\balpha}, \hat{\btheta}) \in \arg \min_{} \mcL_n(g_{\balpha, \btheta}) = 
\arg\min_{(\balpha, \btheta)}\left[ \mcR_n(g_{\balpha, \btheta}) + \lambda_1\|\balpha\|_1 +\lambda_{2}\sum_{l=1}^{\mcD}\lVert \boldsymbol{W}_l-\boldsymbol{I}\rVert_{1} \right].
\] 

 \textbf{Covering number}. Given a $\delta$-uniform covering of $\mcG$, we denote the centers of the balls by $g_q, q =1, \cdots, \mcN_{2n}$, where $\mcN_{2n} = \sup_{\bx}\mcN_{2n}\left(\delta, \|\cdot\|_{\infty}, \mcG|_{\bx} \right)$ is the uniform covering number with
radius $\delta$ under the norm $\|\cdot\|_{\infty}$. By the definition of covering, there exists a $q^*$ such that
$\|g_{\hat{\balpha}, \hat{\btheta}} - g_{q^*}\|_{\infty} \leq \delta$ on $\bx \in (\bX_1, \cdots, \bX_n, \bX_1', \cdots, \bX_{n}')$. Let $A \preceq B$ represent $A \leq c B$ for a postive constant $c$.

\noindent \textbf{Proof of Theorem~\ref{th:cons}}. 
Let $(\check\balpha, \check\btheta) \in \bTheta$ such that 
\[
(\check{\balpha}, \check{\btheta}) \in \arg\min_{(\balpha, \btheta) \in \bTheta} \mbE\bigg[d\bigg(\big(\hat{\balpha}, \hat{\btheta}\big), \bTheta \bigg)\bigg].
\]
Without loss of generality, we slightly abuse notation by writing $\sum_{l=1}^{\mcD}\lVert \boldsymbol{W}_l-\boldsymbol{I}\rVert_{1}$ to denote $\sum_{l=1}^{\mcD}(\lVert \boldsymbol{W}_l-\boldsymbol{I}\rVert_{1}+|c_l|)$, where the intercept term is absorbed into the weight matrix $\boldsymbol{W}_l$ of the neural network. Then, it follows that 
\[
\mcR_n(g_{\hat{\balpha}, \hat{\btheta}}) + \lambda_1\|\hat\balpha\|_1 + \lambda_{2}\sum_{l=1}^{\mcD}\lVert \hat{\boldsymbol{W}}_l-\boldsymbol{I}\rVert_{1} \leq 
\mcR_n(g_{\check{\balpha}, \check{\btheta}}) + \lambda_1\|\check\balpha\|_1 +\lambda_{2}\sum_{l=1}^{\mcD}\lVert \check{\boldsymbol{W}}_l-\boldsymbol{I}\rVert_{1}.
\]

For the expected excess risk, we have the following decomposition,
\begin{eqnarray}
\label{eq:t1}
&&\mbE\bigg(
\mcR(g_{\hat\balpha, \hat\btheta}) - \mcR(g_{\check\balpha, \check\btheta})\bigg) \nonumber\\
&= &\mbE\left( \mcR(g_{\hat\balpha, \hat\btheta}) - \mcR_n(g_{\hat\balpha, \hat\btheta}) \right) + 
\mbE\left( \mcR_n(g_{\hat\balpha, \hat\btheta}) - \mcR_{n}(g_{\check\balpha, \check\btheta}) \right) \\
&& + \mbE\big( \mcR_{n}(g_{\check\balpha, \check\btheta}) - \mcR(g_{\check\balpha, \check\btheta}) \big) \nonumber\\
&\leq & I_1 + \mbE_{D_n}\left[\lambda_1 \left( \|\check\balpha\|_1 - \|\hat\balpha\|_1\right) + \lambda_2 \sum_{l=1}^{\mcD}
\Big(\|\check{\boldsymbol{W}}_l - \boldsymbol{I}\|_1 - \|\hat{\boldsymbol{W}}_l - \boldsymbol{I}\|_1\Big)\right]. 
\end{eqnarray}

Note that the first term has the following upper bound: 
\begin{eqnarray}
\label{eq:t2}
I_1 &:=& \mbE\left( \mcR(g_{\hat\balpha, \hat\btheta}) - \mcR_n(g_{\hat\balpha, \hat\btheta}) \right) \nonumber\\
&=&\mbE_{D_n}\left\{ n^{-1}\sum_{i=1}^n\left[ 
 \mbE_{D'_n}\left( n^{-1}\sum_{i=1}^n \left( Y'_i - g_{\hat{\btheta}}(\hat{\balpha} \odot \bX'_i) \right)^2 \right) \right.\right. \\
 &&\left. \left.-  \left( Y_i - g_{\hat{\btheta}}(\hat{\balpha} \odot \bX_i) \right)^2 \right] \right\} \nonumber\\
 & = & \mbE_{D_n}\left\{ n^{-1}\sum_{i=1}^n\left[ 
 \mbE_{D'_n}\left( n^{-1}\sum_{i=1}^n S(g_{\hat{\balpha}, \hat{\btheta}}, \bX_i^{\prime}) \right)
 - 2S(g_{\hat{\balpha}, \hat{\btheta}}, \bX_i^{\prime}) \right] \right\} \nonumber\\
 && + \mbE_{D_n}\left\{ n^{-1}\sum_{i=1}^n \left[ \big(Y_i - g_{\hat\btheta}(\hat\balpha \odot \bX_i)\big)^2 - \big(Y_i - g_{\check{\btheta}}(\check{\balpha} \odot \bX_i)\big)^2 \right] \right\} \nonumber\\
 && + \mbE_{D_n}\left\{ n^{-1}\sum_{i=1}^n S(g_{\check{\btheta}, \check{\balpha}}, \bX_i^{\prime}) \right\} \nonumber\\
 &\leq& I_{11} + \mbE_{D_n}\left[\lambda_1\left( \| \check{\balpha} \|_1 - \| \hat{\balpha} \|_1\right) + \lambda_2 \sum_{l=1}^{\mcD}
\Big(\|\check{\boldsymbol{W}}_l - \boldsymbol{I}\|_1 - \|\hat{\boldsymbol{W}}_l - \boldsymbol{I}\|_1\Big)\right]  \nonumber\\
&&+ \left( \mcR(g_{\check{\balpha}, \check{\btheta}}) - \mcR(g^*) \right).
\end{eqnarray}

Next, we will give an upper bound of $I_{11}$ and handle it with truncation and classical chaining technique of 
empirical processes. 
According to the definition of $S(g_{\balpha, \btheta}, \bX_i)$, we have 
\begin{eqnarray*}
&&\big| S(g_{\hat\balpha, \hat\btheta}, \bX_i^{\prime}) - S(g_{q^*}, \bX_i^{\prime})\big| \\
&=& 
\big| 2Y_i \big(g_{\hat\btheta}(\hat{\balpha} \odot \bX_i) - g_{q^*}(\bX_i)\big) + \big(g^2_{\hat\btheta}(\hat{\balpha} \odot \bX_i) - g^2_{q^*}(\bX_i)\big)  \big|\\
&\leq& (2|Y_i| + 2\mcB)\delta,
\end{eqnarray*}
and 
\begin{eqnarray*}
&&\big| L(g_{\hat\balpha, \hat\btheta}, \bX_i^{\prime}) - L(g_{q^*}, \bX_i^{\prime})\big| \\
&\leq& 
 \mbE_{D'_n}\left(\big| S(g_{\hat\balpha, \hat\btheta}, \bX_i^{\prime}) - S(g_{q^*}, \bX_i^{\prime})\big| \right) + 2\big| S(g_{\hat\balpha, \hat\btheta}, \bX_i^{\prime}) - S(g_{q^*}, \bX_i^{\prime})\big| \\
 &\leq& 3 (\mbE|Y_i| + \mcB + |Y_i|) \delta.
\end{eqnarray*}
Then, it follows that 
\begin{eqnarray*}
\mbE_{D_n}\left( n^{-1}\sum_{i=1}^n \big| L(g_{\hat\balpha, \hat\btheta}, \bX_i^{\prime}) - L(g_{q^*}, \bX_i^{\prime}) \big| \right) \leq 
6(\mbE|Y_i| + \mcB) \delta,
\end{eqnarray*}
which leads to 
\begin{eqnarray}
\label{eq:lc}
\mbE_{D_n}\left(n^{-1}\sum_{i=1}^n L(g_{\hat\balpha, \hat\btheta}, \bX_i^{\prime})  \right) \leq 
\mbE_{D_n}\left(n^{-1}\sum_{i=1}^n L(g_{q^*}, \bX_i^{\prime})  \right) + 6(\mbE|Y_i| + \mcB) \delta.
\end{eqnarray}

Let $0 < \beta_n$ be a positive number who may depend on the sample size $n$. Denote 
$T_{\beta_n} Y = Y$ if $Y \leq \beta_n$ and $T_{\beta_n} Y = \beta_n$ otherwise. Define the function 
$g^*_{\beta_n}$ by 
\[
g^*_{\beta_n}(\bx) = \arg\min_{g: \|g\|_{\infty} < \mcB} \mbE\left[\big(T_{\beta_n}Y_i - g(\bX_i) \big)^2 \mid \bX_i = \bx \right].
\]
For any $g \in \mcG(\balpha, \btheta)$, let 
$S_{\beta_n}(g_{\balpha, \btheta}, \bX_i^{\prime}) = \big(T_{\beta_n}Y_i - g_{\btheta}(\balpha \odot \bX_i) \big)^2 - 
\big(T_{\beta_n}Y_i - g^*_{\beta_n}(\bX_i) \big)^2
$. Then, we have 
\begin{eqnarray*}
\mbE\big(S(g_{\alpha, \btheta}, \bX_i^{\prime}) \big) &=& 
\mbE\big(S_{\beta_n}(g_{\alpha, \btheta}, \bX_i^{\prime}) \big) + 
\mbE\left[\big(Y_i - g_{\btheta}(\balpha \odot \bX_i) \big)^2 - \big(T_{\beta_n}Y_i - g_{\btheta}(\balpha \odot \bX_i) \big)^2 \right] \\
&&- \mbE\left[\big(Y_i - g^*(\bX_i) \big)^2 - \big(T_{\beta_n}Y_i - g^*(\bX_i) \big)^2 \right] \\
&& - \mbE\left[\big(T_{\beta_n}Y_i - g^*(\bX_i) \big)^2 - \big(T_{\beta_n}Y_i - g^*_{\beta_n}(\bX_i) \big)^2 \right] \\
&\leq & \mbE\big(S_{\beta_n}(g_{\alpha, \btheta}, \bX_i^{\prime}) \big)  + 
\mbE\big(Y_i^2 - (T_{\beta_n}Y_i)^2 - 2(Y_i - T_{\beta_n}Y_i)g_{\btheta}(\balpha \odot \bX_i)\big) \\
&& - \mbE\big(Y_i^2 - (T_{\beta_n}Y_i)^2 - 2(Y_i - T_{\beta_n}Y_i)g^*(\bX_i)\big)\\
&\leq &  \mbE\big(S_{\beta_n}(g_{\alpha, \btheta}, \bX_i^{\prime}) \big)  + 4 \mbE(|Y_i|I(Y_i > \beta_n))\mcB,
\end{eqnarray*}
and 
\begin{eqnarray*}
\mbE\big(S_{\beta_n}(g_{\alpha, \btheta}, \bX_i^{\prime}) \big) &=& 
\mbE\big(S(g_{\alpha, \btheta}, \bX_i^{\prime}) \big) -
\mbE\left[\big(Y_i - g_{\btheta}(\balpha \odot \bX_i) \big)^2 - \big(T_{\beta_n}Y_i - g_{\btheta}(\balpha \odot \bX_i) \big)^2 \right] \\
&&- \mbE\left[\big(T_{\beta_n}Y_i - g^*_{\beta_n}(\bX_i) \big)^2 - \big(Y_i - g^*_{\beta_n}(\bX_i) \big)^2 \right] \\
&& - \mbE\left[\big(Y_i - g^*_{\beta_n}(\bX_i) \big)^2 - \big(Y_i - g^*(\bX_i) \big)^2 \right] \\
&\leq & \mbE\big(S(g_{\alpha, \btheta}, \bX_i^{\prime}) \big)  -
\mbE\big(Y_i^2 - (T_{\beta_n}Y_i)^2 - 2(Y_i - T_{\beta_n}Y_i)g_{\btheta}(\balpha \odot \bX_i)\big) \\
&& + \mbE\big(Y_i^2 - (T_{\beta_n}Y_i)^2 - 2(Y_i - T_{\beta_n}Y_i)g^*_{\beta_n}(\bX_i)\big)\\
&\leq &  \mbE\big(S(g_{\alpha, \btheta}, \bX_i^{\prime}) \big)  + 4 \mbE(|Y_i|I(Y_i > \beta_n))\mcB,
\end{eqnarray*}
which leads to 
\[
\bigg| \mbE\big( S(g_{\alpha, \btheta}, \bX_i^{\prime}) -S_{\beta_n}(g_{\alpha, \btheta}, \bX_i^{\prime}) \big)\bigg|
\leq 4 \mbE(|Y_i|I(Y_i > \beta_n))\mcB. 
\]
Then, 
\begin{eqnarray}
\label{eq:lt}
&&\bigg| \mbE_{D_n}\left[
\frac{1}{n}\sum_{i=1}^n\big( L(g_{q^*}, \bX_i^{\prime}) - L_{\beta_n}(g_{q^*}, \bX_i^{\prime})  \big)
\right] \bigg|  \nonumber\\
&\leq& 
\big| \mbE_{D_n'}\big( S(g_{q^*}, \bX_i^{\prime}) - S_{\beta_n}(g_{q^*}, \bX_i^{\prime}) \big) \big| + 
2\big| \mbE_{D_n} \big( S(g_{q^*}, \bX_i^{\prime}) - S_{\beta_n}(g_{q^*}, \bX_i^{\prime}) \big) \big| \nonumber\\
&\leq&12 \mbE(|Y_i|I(Y_i > \beta_n))\mcB.
\end{eqnarray}

On the other hand, for any $g \in \mcG(\balpha, \btheta)$, we have 
\begin{eqnarray*}
&&|S_{\beta_n}(g, \bX_i^{\prime})| \leq 5(\beta_n + \mcB)^2, \\
&&\sigma^2_{S}(g) := \text{Var}(S_{\beta_n}(g, \bX_i^{\prime})) \leq \mbE\{S^2_{\beta_n}(g, \bX_i^{\prime})\} \leq 
5(\beta_n + \mcB)^2 \mbE(S_{\beta_n}(g, \bX_i^{\prime})).
\end{eqnarray*}
Following the Bernstein inequality, for any $t > 0$, let $u = t/2 + \sigma^2_S(g)/(10(\beta_n + \mcB)^2)$, we have 
\begin{eqnarray*}
&&P\left\{n^{-1}\sum_{i=1}^n L_{\beta_n}(g_{q}, \bX_i^{\prime}) > t \right\}\\
&=& P\left\{ \mbE_{D_n'}\left(S_{\beta_n}(g_{q}, \bX_i^{\prime}) \right)
-n^{-1}2\sum_{i=1}^nS_{\beta_n}(g_q, \bX_i^{\prime}) > t
\right\}\\
    &=& P\left\{  \mbE_{D_n'}\{S_{\beta_n}(g_q, \bX_i^{\prime})\} - \frac{1}{n}\sum_{i=1}^n S_{\beta_n}(g_q, \bX_i^{\prime}) > \frac{t}{2} + \frac{1}{2} \mbE_{D_n'}\{S_{\beta_n}(g_q, \bX_i^{\prime})\} \right\} \nonumber\\
    &\leq& P\left\{  \mbE_{D_n'}\{S_{\beta_n}(g_q, \bX_i^{\prime})\} - \frac{1}{n}\sum_{i=1}^n S_{\beta_n}(g_q, \bX_i^{\prime}) > \frac{t}{2} + \frac{1}{2} \frac{\sigma_S^2(g)}{5(\beta_n + \mcB)^2} \right\} \nonumber\\
    &\leq& \exp ( -\frac{n u^2}{2\sigma_S^2(g)+20u(\beta_n + \mcB)^2/3} ) \nonumber\\
    &\leq& \exp ( -\frac{n u^2}{20u(\beta_n + \mcB)^2 + 20u(\beta_n + \mcB)^2/3} ) \nonumber\\
    &\leq& \exp ( -\frac{1}{20 + 20/3}\frac{n u}{(\beta_n + \mcB)^2} ) \nonumber\\
    &\leq& \exp ( -\frac{1}{40 + 40/3}\frac{n t}{(\beta_n + \mcB)^2} ) \nonumber\\
    &=& \exp\left( - \frac{Cnt}{(\beta_n + \mcB)^2} \right).
\end{eqnarray*}
This leads to a tail probability bound of $n^{-1}\sum_{i=1}^n L_{\beta_n}(g_{q^*}, \bX_i^{\prime})$, that is, 
\[
P\left\{n^{-1}\sum_{i=1}^n L_{\beta_n}(g_{q^*}, \bX_i^{\prime}) > t \right\} \leq 2 \mcN_{2n}\exp\left( - \frac{Cnt}{(\beta_n + \mcB)^2} \right).
\]
Then for $a_n>0$,
\begin{eqnarray*}
\mbE_{D_n}\left[ \frac{1}{n}\sum_{i=1}^nL_{\beta_n}(g_{q^*}, \bX_i^{\prime}) \right]
&\leq& a_n + \int_{a_n}^{\infty}P\left\{ \frac{1}{n}\sum_{i=1}^nL_{\beta_n}(g_{q^*}, \bX_i^{\prime}) > t \right\}dt \nonumber\\
&\leq& a_n + \int_{a_n}^{\infty}2\mcN_{2n}\exp\left( - \frac{Cnt}{(\beta_n + \mcB)^2} \right)dt \nonumber\\
&\leq& a_n + 2\mcN_{2n}\exp\left( -a_n \frac{Cn}{(\beta_n + \mcB)^2} \right)\frac{(\beta_n + \mcB)^2}{Cn}.
\end{eqnarray*}
Choosing $a_n=\log 2\mcN_{2n}\frac{(\beta_n + \mcB)^2}{Cn}$, the above inequality leads to
\begin{eqnarray}
\label{eq:i1eq3-o}
 \mbE_{D_n}\left[ \frac{1}{n}\sum_{i=1}^nL_{\beta_n}(g_{q^*}, \bX_i^{\prime}) \right] \leq \frac{C(\beta_n + \mcB)^2 (\log 2\mcN_{2n} + 1) }{n}.
\end{eqnarray}

Combining inequalities~\eqref{eq:lc}, \eqref{eq:lt}, and~\eqref{eq:i1eq3-o}, we have 
\begin{eqnarray*}
I_{11} &=& \mbE_{D_n}\left(n^{-1}\sum_{i=1}^n L(g_{\hat{\balpha}, \hat{\btheta}}, \bX_i^{\prime}) \right)\\
&\leq& \mbE_{D_n}\left(n^{-1}\sum_{i=1}^n L(g_{q^*}, \bX_i^{\prime}) \right) + 6 \left( \mbE|Y_i| + \mcB\right)\delta \\
&\leq& \mbE_{D_n}\left(n^{-1}\sum_{i=1}^n L_{\beta_n}(g_{q^*}, \bX_i^{\prime}) \right) + 6 \left( \mbE|Y_i| + \mcB\right)\delta + 
12\mbE(|Y_i|I(Y_i > \beta_n))\mcB \\
&\leq & \frac{C(\beta_n + \mcB)^2 (\log 2\mcN_{2n} + 1) }{n} + 6 \left( \mbE|Y_i| + \mcB\right)\delta + 
12\mbE(|Y_i|I(Y_i > \beta_n))\mcB.
\end{eqnarray*}

Let $\beta_n = \log n$ and $\delta = n^{-1}$. Under conditions~(C2) and (C3), using the above inequalities, we obtain that 
\begin{eqnarray}
\label{eq:i11}
I_{11} \preceq \frac{C\log^2 n \log \mcN_{2n}}{n}.
\end{eqnarray}

Then, combining inequalities~\eqref{eq:t1}, \eqref{eq:t2}, and~\eqref{eq:i11}, using the condition that $\mbE\bigg[d\bigg((\balpha, \btheta), \bTheta\bigg) \bigg]
\leq c \mbE\bigg(
\mcR(g_{\balpha, \btheta}) - \mcR(g_{\check{\balpha}, \check{\btheta}})\bigg)$, we can obtain 
\begin{equation}
\begin{aligned}
\label{eq:if1}
&\mbE\!\left[d\!\left((\hat\balpha,\hat\btheta),\bTheta\right)\right]
= \mbE\!\left(\big\|(\hat\balpha,\hat\btheta)-(\check\balpha,\check\btheta)\big\|_2^2\right)
   \;\le\; c\,\mbE\!\left\{\mcR\!\big(g_{\hat\balpha,\hat\btheta}\big)-\mcR\!\big(g_{\check\balpha,\check\btheta}\big)\right\}\\
&\le c\, I_1
   + c\,\mbE\!\left[
       \lambda_1\big(\|\check\balpha\|_1-\|\hat\balpha\|_1\big)
       + \lambda_2 \sum_{l=1}^{\mcD} \big(\|\check{\boldsymbol{W}}_l-\boldsymbol{I}\|_1-\|\hat{\boldsymbol{W}}_l-\boldsymbol{I}\|_1\big)
     \right] \\
&\le c\, I_{11}
   + c\,\lambda_1\,\mbE\!\left[\|\check\balpha-\hat\balpha\|_1\right]
   + c\,\lambda_2 \sum_{l=1}^{\mcD} \mbE\!\left[\|\check{\boldsymbol{W}}_l-\hat{\boldsymbol{W}}_l\|_1\right]\\
&\le c\,\frac{C\,\log^2 n \,\log \mcN_{2n}(n^{-1},\|\cdot\|_\infty,\mcG|_{\bx})}{n}
   + c\,\lambda_1 \sqrt{d}\,\mbE\!\left[\|\check\balpha-\hat\balpha\|_2\right] \\
&\qquad
   + c\,\lambda_2 \sum_{l=1}^{\mcD} \mcW \,\mbE\!\left[\|\check{\boldsymbol{W}}_l-\hat{\boldsymbol{W}}_l\|_F\right], 
\end{aligned}
\end{equation}

To let the right side of the above inequality to be minimum, we have that 
\begin{equation}
\label{eq:lam2}
\begin{aligned}
 &   \lambda_1^2  \le \frac{c\log^2 n \,\log \mcN_{2n}}{nd}
\;+\;\frac{c\big(\mcR(g_{\check\balpha,\check\btheta})-\mcR(g^*)\big)}{d},\\
&\lambda_2^2 \le \frac{c\log^2 n \,\log \mcN_{2n}}{n\mcW^2\mcD}
\;+\;\frac{c\big(\mcR(g_{\check\balpha,\check\btheta})-\mcR(g^*)\big)}{\mcW^2\mcD},
\end{aligned}
\end{equation}

Note that using Young's inequality, we have 

\begin{equation}
\label{eq:ip}
    \begin{aligned}
&\mbE\!\left[c\lambda_1\sqrt{d}\,\|\check\balpha-\hat\balpha\|_2\right]
\;\le\;
\frac{1}{2}\,\mbE\!\left[\|\check\balpha-\hat\balpha\|_2^2\right]
\;+\;\frac{c^2\lambda_1^2 d}{2}. \\
&\mbE\!\left[c\lambda_2\,\mcW\,\|\check{\boldsymbol{W}}_l-\hat{\boldsymbol{W}}_l\|_F\right]
\;\le\;
\frac{1}{2}\,\mbE\!\left[\|\check{\boldsymbol{W}}_l-\hat{\boldsymbol{W}}_l\|_F^2\right]
\;+\;\frac{c^2\lambda_2^2 \mcW^2}{2}.
\end{aligned}
\end{equation}

Combining inequalities~\eqref{eq:if1}, \eqref{eq:lam2} and~\eqref{eq:ip}, we can obtain that 
\[
\begin{aligned}
&\mbE\!\left[d\!\left((\hat\balpha,\hat\btheta),\bTheta\right)\right]\\
\;\le\;&
c\left\{
\frac{\log^2 n \,\log \mcN_{2n}}{n}
\;+\;\big(\mcR(g_{\check\balpha,\check\btheta})-\mcR(g^*)\big)
\;+\;\lambda_1^2 d
\;+\;\lambda_2^2 \mcW^2 \mcD
\right\}.\\
\;\le\;
&c\left\{
\frac{\log^2 n\,\log \mcN_{2n}}{n}
\;+\;\big(\mcR(g_{\check\balpha,\check\btheta})-\mcR(g^*)\big)
\right\}.
\end{aligned}
\]
\color{black}

Let $(\balpha^*_s, \btheta^*_s) \in \arg\min_{\balpha_s, \btheta_s}\mcR(g_{\balpha_s, \btheta_s}) = \arg\min_{\balpha_s, \btheta_s}\mbE\left( Y - g_{\btheta_s}(\balpha_s \odot \bX_s)\right)^2$. Define $\bTheta_s = \{(\balpha_s, \btheta_s): \mcR(g_{\balpha_s, \btheta_s}) = \mcR(g_{\balpha_s^*, \btheta^*_s})\}$.
Define $d((\balpha_s, \btheta_s), \bTheta_s) = \min_{(\balpha_s, \btheta_s) \in \bTheta_s} \| (\hat{\balpha}_s, \hat{\btheta}_s) - (\balpha_s, \btheta_s) \|_2^2$, and 
write $(\check{\balpha}_s, \check{\btheta}_s) \in \arg\min_{(\balpha_s, \btheta_s) \in \bTheta_s} \| (\hat{\balpha}_s, \hat{\btheta}_s) - (\balpha_s, \btheta_s) \|_2^2$, where 
$\check\btheta_s = \{(\check{W}_{s, l}, \check{c}_{s, l}), l = 0, \cdots, L\}$. Denote $\tilde{\btheta} \in \arg\min_{\btheta} \mcR(g_{\tilde{\balpha}, \btheta})$, where $\tilde{\balpha} = (\check{\balpha}_s, \bm{0})$, 
Then, it follows from Lemma~\ref{th:le1} that $(\check{\balpha}_s, \bm{0}, \tilde{\btheta}) \in \bTheta$.

It follows from
\[
\mcR_n(g_{\hat{\balpha}, \hat{\btheta}}) + \lambda_1\|\hat{\balpha}\|_1 +\lambda_{2}\sum_{l=1}^{\mcD}\lVert \hat{\boldsymbol{W}}_l-\boldsymbol{I}\rVert_{1}  \leq \mcR_n(g_{\check{\balpha}_s, \bm{0}, \check{\btheta}}) + \lambda_1\| \check{\balpha}_s \|_1 + \lambda_{2}\sum_{l=1}^{\mcD}\lVert \check{\boldsymbol{W}}_l-\boldsymbol{I}\rVert_{1} , 
\]
that 
\begin{align}
\lambda_1 \|\hat{\balpha}_c\|_1
&\le \mcR_n(g_{\check{\balpha}_s, \bm{0}, \check{\btheta}}) - \mcR_n(g_{\hat{\balpha}, \hat{\btheta}})
+ \lambda_1 \big( \| \check{\balpha}_s \|_1 - \|\hat{\balpha}_s \|_1 \big)
+ \lambda_{2}\!\sum_{l=1}^{\mcD}\!\Big(\big\lVert \check{\boldsymbol{W}}_l-\boldsymbol{I}\big\rVert_{1}
 - \big\lVert \hat{\boldsymbol{W}}_l-\boldsymbol{I}\big\rVert_{1}\Big).
\label{eq:step-start}
\end{align}

Insert and subtract population risks gives:
\begin{align}
\lambda_1 \|\hat{\balpha}_c\|_1
&\le \Big[\mcR(g_{\check{\balpha}_s, \bm{0}, \check{\btheta}}) - \mcR(g_{\hat{\balpha}, \hat{\btheta}})\Big]
+ \Big[\mcR(g_{\hat{\balpha}, \hat{\btheta}}) - \mcR_n(g_{\hat{\balpha}, \hat{\btheta}})\Big] \nonumber\\
&\quad + \lambda_1 \big( \| \check{\balpha}_s \|_1 - \|\hat{\balpha}_s \|_1 \big)
+ \lambda_{2}\!\sum_{l=1}^{\mcD}\!\Big(\big\lVert \check{\boldsymbol{W}}_l-\boldsymbol{I}\big\rVert_{1}
 - \big\lVert \hat{\boldsymbol{W}}_l-\boldsymbol{I}\big\rVert_{1}\Big).
\label{eq:insert-pop}
\end{align}
Because $(\check{\balpha}_s,\bm{0},\check{\btheta})\in\Theta$ we have $\mcR(g_{\check{\balpha}_s, \bm{0}, \check{\btheta}})=\mcR(g^*_{\mcG})\le \mcR(g_{\hat{\balpha}, \hat{\btheta}})$, hence the first bracket in \eqref{eq:insert-pop} is nonpositive and can be dropped:
\begin{align}
\lambda_1 \|\hat{\balpha}_c\|_1 
\le &\Big[\mcR(g_{\hat{\balpha}, \hat{\btheta}}) - \mcR_n(g_{\hat{\balpha}, \hat{\btheta}})\Big]
+ \lambda_1 \big( \| \check{\balpha}_s \|_1 - \|\hat{\balpha}_s \|_1 \big) \notag\\
&+ \lambda_{2}\!\sum_{l=1}^{\mcD}\!\Big(\big\lVert \check{\boldsymbol{W}}_l-\boldsymbol{I}\big\rVert_{1}
 - \big\lVert \hat{\boldsymbol{W}}_l-\boldsymbol{I}\big\rVert_{1}\Big).
\label{eq:after-drop}
\end{align}

For the gate term, $\|\check\balpha_s\|_1-\|\hat\balpha_s\|_1 \le \|\check\balpha_s-\hat\balpha_s\|_1 \le \sqrt{d}\,\|\check\balpha_s-\hat\balpha_s\|_2$.
For the depth term, by triangle inequality and $\ell_1$–$\ell_2$,
\[
\sum_{l=1}^{\mcD}\!\Big|\big\lVert \check{\boldsymbol{W}}_l-\boldsymbol{I}\big\rVert_{1}
 - \big\lVert \hat{\boldsymbol{W}}_l-\boldsymbol{I}\big\rVert_{1}\Big|
\;\le\;
\sum_{l=1}^{\mcD}\!\big\lVert \check{\boldsymbol{W}}_l-\hat{\boldsymbol{W}}_l\big\rVert_{1}
\;\le\; \sqrt{S}\,\|\check{\btheta}-\hat{\btheta}\|_2,
\]
where $S = \mcW^2 \mcD $ is the total number of scalar parameters in $\theta$.
Therefore \eqref{eq:after-drop} yields
\begin{align}
\lambda_1 \|\hat{\balpha}_c\|_1
&\le \Big[\mcR(g_{\hat{\balpha}, \hat{\btheta}}) - \mcR_n(g_{\hat{\balpha}, \hat{\btheta}})\Big]
+ \lambda_1 \sqrt{d}\,\|\check\balpha_s-\hat\balpha_s\|_2
+ \lambda_2 \sqrt{S}\,\|\check{\btheta}-\hat{\btheta}\|_2.
\label{eq:l1-bound}
\end{align}

With Young's inequality, we have:
\begin{align}
\lambda_1 \|\hat{\balpha}_c\|_1
&\le \Big[\mcR(g_{\hat{\balpha}, \hat{\btheta}}) - \mcR_n(g_{\hat{\balpha}, \hat{\btheta}})\Big]
+ \tfrac12\|\check\balpha_s-\hat\balpha_s\|_2^2 + \tfrac12\lambda_1^2 d
+ \tfrac12\|\check{\btheta}-\hat{\btheta}\|_2^2 + \tfrac12\lambda_2^2 S.
\label{eq:young}
\end{align}

Taking expectations:
\begin{align}
\mbE\!\left[\lambda_1 \|\hat{\balpha}_c\|_1\right]
&\le \mbE\!\left[\mcR(g_{\hat{\balpha}, \hat{\btheta}}) - \mcR_n(g_{\hat{\balpha}, \hat{\btheta}})\right]
+ \tfrac12\,\mbE\!\left[\|(\hat\balpha_s,\hat\btheta)-(\check\balpha_s,\check\btheta)\|_2^2\right]
+ \tfrac12\lambda_1^2 d + \tfrac12\lambda_2^2 S.
\label{eq:final-exp}
\end{align}

which leads to 
\[
\mbE(\|\hat{\balpha}_c\|^2_2) \leq  \mbE(\|\hat{\balpha}_c\|^2_1) \leq \frac{c\log^2 n \log \mcN_{2n} }{n} + c \left( \mcR(g_{\check{\balpha}_s, \bm{0}, \check{\btheta}}) - \mcR(g^*) \right).
\]

\subsection*{Lemmas}
\begin{lemma}
\label{th:le1}
Under condition (C1), if $p = O(log^c n)$ for some positive constant $c$, for 
$(\balpha^*, \btheta^*) \in \arg\min_{g \in \mcG(\balpha, \btheta)} \mbE(Y_i - g(\bX_i))^2$, then it follows that 
\begin{itemize}
\item[{(i)}] $|\alpha^*_j| \geq \tau_d$, for $\forall j = 1, \cdots, d$ and some positive constant $c$; 
\item[{(ii)}] there exists a solution $(\balpha^*, \btheta^*)$, such that $\balpha^*_c = 0$.  
\end{itemize}
\end{lemma}

\noindent \textbf{Proof of Lemma~\ref{th:le1}}. 

(i) Suppose that there exists a $(\balpha^*, \btheta^*)$ such that 
$|\alpha^*_j| < c\tau_d$ for at least one $j \in \{1, \cdots, d\}$. Then, for any random vector $\bX \in [0, 1]^p$, we construct 
the vector $\bX_{[j]}(0) = (X_{[1]}, \cdots, X_{[j-1]}, 0, X_{[j+1]}, \cdots, X_{[p]})$, clearly, 
$|g_{\btheta^*}\big(\balpha^* \odot \bX_{[j]}(0)\big) - g_{\btheta^*}\big( \balpha^* \odot \bX)| \leq c\tau_d$ for some positive constant $c$. Based on the definition that 
$g^*_{\mcG} = \arg\min_{g \in \mcG(\balpha, \btheta)}\mbE(Y_i - g_{\btheta}(\balpha \odot \bX_i))^2$, that 
\[
g_{\btheta^*}\big(\balpha^* \odot \bX_{[j]}(0)\big) = \mbE(Y \mid \bX_{[-j]} = \bx_{[-j]}), \; \; \text{and} \; \; 
g_{\btheta^*}(\balpha^* \odot \bX) = \mbE(Y \mid \bX = \bx).
\]
It contradicts condition (C2). Thus, for any $j = 1, \cdots, d$, $|\alpha^*_j| \geq \tau_d$ for some positive constant $c$. 

(ii) Let $(\balpha^*_s, \btheta^*_s) \in \arg\min_{\balpha_s, \btheta_s}\mcR(g_{\balpha_s, \btheta_s}) = \arg\min_{\balpha_s, \btheta_s}\mbE\left( Y - g_{\btheta_s}(\balpha_s \odot \bX_s)\right)^2$. Define $\bTheta_s = \{(\balpha_s, \btheta_s): \mcR(g_{\balpha_s, \btheta_s}) = \mcR(g_{\balpha_s^*, \btheta^*_s})\}$.
Define $d((\balpha_s, \btheta_s), \bTheta_s) = \min_{(\balpha_s, \btheta_s) \in \bTheta_s} \| (\hat{\balpha}_s, \hat{\btheta}_s) - (\balpha_s, \btheta_s) \|_2^2$, and 
write $(\check{\balpha}_s, \check{\btheta}_s) \in \arg\min_{(\balpha_s, \btheta_s) \in \bTheta_s} \| (\hat{\balpha}_s, \hat{\btheta}_s) - (\balpha_s, \btheta_s) \|_2^2$, where 
$\check\btheta_s = \{(\check{A}_{s, l}, \check{c}_{s, l}), l = 0, \cdots, L\}$. Denote $\tilde{\btheta} \in \arg\min_{\btheta} \mcR(g_{\tilde{\balpha}, \btheta})$, where $\tilde{\balpha} = (\check{\balpha}_s, \bm{0})$. Let $\check{\btheta} = \{(\tilde{A}_l, \tilde{c}_l), l = 0, \cdots, L\}$ with $\tilde{c}_ l = \check{c}_{s, l}$ for $l = 0, \cdots, L$, $\tilde{A}_l = \check{A}_{s, l}$ for $l = 1, \cdots, L$, and 
$\tilde{A}_0 = (\check{A}_{s, 0}, A_{c, 0})$. Then, based on the definition~\eqref{eq:imp}, it is easy to show that 
\begin{eqnarray*}
\mcR(g_{\tilde{\balpha}, \tilde{\btheta}}) \leq \mcR(g_{\tilde{\balpha}, \check{\btheta}}) = \mcR(g_{\check\balpha_s, \check\btheta_s}) = \mcR(g_{\balpha^*, \btheta^*}).
\end{eqnarray*}
Clearly, $(\tilde{\balpha}, \tilde{\btheta})$ is the solution.

\begin{lemma}(Approximation error, \cite[Theorem 3.3 in ][]{jiao2023deep})
\label{lem1}

Given $H\ddot o lder$ smooth functions $g^* \in \mathcal{H}_\beta([0,1]^d,B_0)$, for any $D \in \mathbb{N}^+$ and $W \in \mathbb{N}^+$,  there exists a function $g^*_{\mcG}$ implemented by a ReLU feedforward neural network with width $\mathcal{W} = 38 (\left \lfloor \beta \right \rfloor +1)^{2}d^{\left \lfloor \beta \right \rfloor+1}W\left \lceil \log_2(8W) \right \rceil$ and depth $\mathcal{D} =21(\left \lfloor \beta \right \rfloor +1)^{2}D\left \lceil \log_2(8D) \right \rceil$ such that
\begin{equation*}
	\begin{split}
	\left |g^* - g_{\mcG}^*\right | \leq 18 B_0(\left \lfloor \beta \right \rfloor +1)^{2}d^{\left \lfloor \beta \right \rfloor+\max\{ \beta, 1 \}/2}(WD)^{-2\beta/d},
	\end{split}
\end{equation*}
for all $x \in \left [ 0,1 \right ]^{d} \setminus \Omega(\left [ 0,1 \right ]^{d},K,\delta) $ where
\begin{equation*}
	\begin{split}
	\Omega(\left [ 0,1 \right ]^{d},K,\delta)=\bigcup_{i=1}^{d}\{ x=\left [ x_1,\cdots,x_d \right ]^{T}: x_i \in \bigcup_{k=1}^{K-1}(k/ K-\delta,k / K) \},
	\end{split}
\end{equation*}
with $K=\left \lfloor WD \right \rfloor$ and $\delta$ is an arbitrary number in $(0,1 / (3K)]$.
\end{lemma}

\begin{lemma}(Bounding the covering number, \cite[Theorem 12.2 in ][]{anthony1999neural} and \cite[Theorems 3 and 7 in][]{bartlett2019nearly})
\label{lem2}

Let ReLU feedforward neural network $\mcG$ be a set of real functions from a domain $\mcX$ to the bounded interval $[0, \mcB]$. There exists a universal constant $C$ such that the following holds. Given any $\mathcal{D}, \mathcal{S}$ with $\mathcal{S} > C\mathcal{D} > C^2$, there exists network class $\mcG$ with $\leq \mathcal{D}$ layers and $\leq \mathcal{S}$ parameters with $VC$-dimension $\geq \mathcal{S}\mathcal{D} \log (\mathcal{S}/\mathcal{D})/C$ and given $\delta > 0$
\begin{equation*}
	\begin{split}
	\sup_{\bx}\log \mcN_{2n}(\delta, \|\cdot\|_{\infty}, \mcG|_{\bx}) = O\big( \mathcal{S}\mathcal{D} \log (\mathcal{S}/\delta) \big).
	\end{split}
\end{equation*}
\end{lemma}

\end{document}